\documentclass{article}%

\usepackage{graphicx}
\usepackage{multicol,multirow}
\usepackage{amsmath,amssymb,amsfonts}
\usepackage{mathrsfs}
\usepackage{amsthm}
\usepackage[margin=1in]{geometry}
\usepackage[numbers]{natbib}
\usepackage[T1]{fontenc}
\IfFileExists{newtxmath.sty}{\usepackage{newtxtext}\usepackage{newtxmath}}{\usepackage{mathptmx}}
\usepackage{textcomp}
\usepackage{booktabs}
\usepackage{xcolor}
\definecolor{weblmcolor}{HTML}{1A5FB4}
\usepackage[colorlinks,allcolors=weblmcolor]{hyperref}
\usepackage{physics}
\usepackage{bm}

\newtheorem{theorem}{Theorem}[section]
\newtheorem{lemma}[theorem]{Lemma}
\newtheorem{proposition}[theorem]{Proposition}
\newtheorem{corollary}[theorem]{Corollary}
\theoremstyle{definition}
\newtheorem{definition}[theorem]{Definition}
\newtheorem{remark}[theorem]{Remark}
\newtheorem{example}[theorem]{Example}
\numberwithin{equation}{section}
\newcommand{\keywords}[1]{\par\noindent\textbf{Keywords:} #1}

\newcommand{\R}{\mathbb{R}}
\newcommand{\N}{\mathbb{N}}
\newcommand{\Z}{\mathbb{Z}}
\newcommand{\C}{\mathbb{C}}
\newcommand{\Gr}{\mathrm{Gr}}
\DeclareMathOperator{\Fix}{Fix}
\DeclareMathOperator{\Emb}{Emb}
\newcommand{\Log}{\mathrm{Log}}
\DeclareMathOperator{\diag}{diag}
\DeclareMathOperator{\Sym}{Sym}
\DeclareMathOperator{\Ran}{Ran}
\DeclareMathOperator{\Lie}{Lie}
\DeclareMathOperator{\spec}{spec}
\DeclareMathOperator{\supp}{supp}
\DeclareMathOperator{\Arg}{Arg}
\DeclareMathOperator{\evw}{ev}
\DeclareMathOperator{\dist}{dist}
\DeclareMathOperator{\conv}{conv}
\newcommand{\Lsu}{L^{SU}_{\mathrm{elem}}}
\newcommand{\Loc}{\mathcal L_n}
\newcommand{\Gel}{\mathcal{G}^{SU}_{\mathrm{elem}}}
\newcommand{\GelU}{\mathcal{G}^{U}_{\mathrm{elem}}}
\newcommand{\Gtl}{\mathcal{G}^{SU}_{\mathrm{2lvl}}}
\newcommand{\GtlU}{\mathcal{G}^{U}_{\mathrm{2lvl}}}
\newcommand{\mF}{\mathscr{F}}
\newcommand{\mFU}{\mathscr{F}^{U}}
\newcommand{\typeSU}{\mathbf{m}}
\newcommand{\typeU}{\bm{\mu}}
\newcommand{\typeTwo}{\mathbf{t}}
\DeclareMathOperator{\Ad}{Ad}

\title{Elementary Quantum Gates from Lie Group Embeddings in \texorpdfstring{$U(2^n)$}{U(2^n)}: Spectral Geometry and Intrinsic Factorization Lengths}

\author{%
Antonio Falc\'o\\
Departamento de Matem\'aticas, F\'{\i}sica y Ciencias Tecnol\'ogicas\\
Universidad Cardenal Herrera-CEU, CEU Universities\\
San Bartolom\'e 55, 46115 Alfara del Patriarca (Valencia), Spain\\
\texttt{afalco@uchceu.es}
\and
Daniela Falc\'o--Pomares\\
Grupo de Investigaci\'on Bisite, Universidad de Salamanca\\
Calle Espejo s/n, 37007 Salamanca (Spain)\\
\texttt{dfp99@usal.es}
\and
Hermann G. Matthies\\
Institute of Scientific Computing\\
Technische Universit\"at Braunschweig, \\
Universit\"atsplatz 2, 38106 Braunschweig, Germany\\
\texttt{h.matthies@tu-bs.de}
}

\begin{document}

\maketitle

\begin{abstract}
We study intrinsic factorizations of quantum gates using determinant-one unitaries supported on arbitrary two-dimensional subspaces
of $\C^N$, where $N=2^n$.
For $U\in U(N)$ write $r=\operatorname{rank}(U-I)$.
When $r\ge2$, we characterize all factors that lower the active rank by two through
traceless compressions of the Cayley transform on $\Ran(U-I)$.
For the ordered active eigenangles $\vartheta_j\in(0,2\pi)$, existence is equivalent to
$\vartheta_1+\vartheta_2\le2\pi\le\vartheta_{r-1}+\vartheta_r$.
For $U\in SU(N)$, combining this construction with eigenangle merging realizes the classical
reflection-length formula in two-level form:
$\ell^{SU}(U)=\max\{\nu(U),r-\nu(U)\}$,
where $\nu(U)$ is the positive eigenangle sum divided by $2\pi$.
On the full-rank open set, the age indexes $N-1$ connected components with constant intrinsic length.
We place these factorizations within the spectral description of the dictionaries formed by faithful $SU(2)$ and $U(2)$ embeddings.
For $n\ge2$, these dictionaries are Haar-null of Hausdorff dimensions $N^2-N+1$ and $N^2-N+2$;
the former is compact and the latter dense.
Distances to compact sectors reduce to eigenvalue matching. Further results concern ambient minimality of embedded geodesics
and a first-order sharp spectral bound on coherence production.
\end{abstract}



\section{Introduction}
\label{sec:introduction}

In quantum circuits, elementary gates are usually specified relative to a tensor factorization of the state space
\cite{NielsenChuang2010,BarencoEtAl1995,DiVincenzo1995,DeutschBarencoEkert1995}. Their implementation cost therefore depends on which
subsystems and interactions are available. A complementary question concerns factorizations whose supporting subspaces may be chosen
freely: how many two-level unitary factors are needed to realize a given gate, and how does this number change when each factor is
required to have determinant one? This paper addresses that question within a conjugation-invariant framework of gates defined by Lie
group embeddings.

Fix $n\ge1$ and $N=2^n$. We call a unitary \emph{elementary} if it belongs to a faithful embedded copy of $SU(2)$ or $U(2)$ in $U(N)$,
and write
\[
\Gel(n):=\bigcup_{\phi\in\Emb(SU(2),U(N))}\phi\big(SU(2)\big),
\qquad
\GelU(n):=\bigcup_{\Phi\in\Emb(U(2),U(N))}\Phi\big(U(2)\big).
\]
Here $\Emb$ denotes injective continuous homomorphisms (Definition~\ref{def:embedding}). Both dictionaries are invariant under
conjugation by $U(N)$. They contain the coordinate two-level gates used in Givens synthesis
\cite{NielsenChuang2010,ReckZeilingerBernsteinBertani1994,Vartiainen2004,Mottonen2004} and the tensor-local one-qubit gates, realized by
different representation types. Their two-level subdictionaries are
\[
\Gtl(n)=\{U\in SU(N):\operatorname{rank}(U-I)\le2\},
\qquad
\GtlU(n)=\{U\in U(N):\operatorname{rank}(U-I)\le2\}
\]
(Proposition~\ref{prop:spectral_consequences}(ii)); the supporting planes may vary over the Grassmannian $\Gr_2(\C^N)$.
Proposition~\ref{prop:grassmannian} describes the supporting planes and frame gauge as a Grassmannian bundle;
Remark~\ref{rem:descriptor_pair} gives the frame formula used in the factorization argument.

\paragraph{Intrinsic factorization lengths.}
Let $\ell^U(U)$ be the minimum number of factors from $\GtlU$, and define $\ell^{SU}(U)$ analogously
for $U\in SU(N)$ using $\Gtl$. With $r:=\operatorname{rank}(U-I)$, Theorem~\ref{thm:intrinsic_length} gives
\[
\ell^U(U)=\lceil r/2\rceil.
\]
For determinant-one factors, the rank alone does not determine the answer.
We use the standard age of a unitary matrix: the sum of its eigenangles in $[0,2\pi)$, divided by $2\pi$.
For $U\in SU(N)$ let $\nu(U)$ be its age.
Theorem~\ref{thm:lsu}(ii) gives
\[
\ell^{SU}(U)=\max\{\nu(U),r-\nu(U)\}.
\]
This count is the two-level form of the classical reflection-length theorem of
Djokovi\'c--Malzan \cite[Theorem]{DjokovicMalzan1979}; the translation between the two
normalizations is given under \emph{Classical ingredients and scope} below.
Theorem~\ref{thm:lsu} also gives a direct construction using age subadditivity,
eigenangle merging and Cayley compressions on the active support.
Its part (i) describes all factors attaining a maximal one-step rank reduction,
rather than only their minimal total number.

The construction uses a criterion for one maximal rank reduction on the active support $R:=\Ran(U-I)$.
For $r\ge2$, the restriction $V:=U|_R$ has no eigenvalue $1$ and has Hermitian Cayley transform
\[
C_R=-\tfrac i2(V+I_R)(V-I_R)^{-1}.
\]
Theorem~\ref{thm:lsu}(i) applies to every $U\in U(N)$ with $r\ge2$:
a two-level $SU(2)$ factor lowers the rank by two exactly when $C_R$ has a traceless two-dimensional compression.
For the ordered active eigenangles, the equivalent condition is
\[
\vartheta_1+\vartheta_2\le2\pi\le\vartheta_{r-1}+\vartheta_r.
\]
A maximal reduction forces the supporting plane into $R$.
For $U\in SU(N)$ it changes $(r,\nu)$ to $(r-2,\nu-1)$.
The criterion propagates in the balanced regime $r=2\nu$, but need not propagate otherwise.
Even in the balanced case, no two eigenangles need sum to $2\pi$; the supporting plane may mix eigenspaces.
The Cayley criterion both supplies this step and parameterizes all admissible factors.
On the full-rank open set, part (iii) identifies $N-1$ connected components indexed by $\nu$.
The rank bound is attained exactly when
$\nu(U)\in\{\lfloor r/2\rfloor,\lceil r/2\rceil\}$.
For example,
\[
\ell^{SU}(iI_4)=3,\quad \ell^U(iI_4)=2,\qquad
\ell^{SU}(e^{i\pi/4}I_8)=7,\quad \ell^U(e^{i\pi/4}I_8)=4.
\]
These lengths count factors with freely chosen supports; they do not include the cost of realizing those supports
in a prescribed circuit architecture.

\paragraph{Spectral description and geometry of the dictionaries.}
The representation types organize the full elementary dictionaries. For an $SU(2)$-type $\typeSU=(m_d)_{d\ge0}$, a unitary belongs to
its sector exactly when, for some single angle $a\in[0,\pi]$, its
eigenvalue multiset is
\[
\biguplus_{d\ge0} m_d\cdot\{e^{i(d-2j)a}\}_{j=0}^{d}.
\]
Here $\biguplus$ denotes multiset union and $m_d\cdot$ means that
the indexed list is repeated $m_d$ times; eigenvalues are counted
with multiplicity (Theorem~\ref{thm:spectral}). Theorem~\ref{thm:spectral_U2} gives the corresponding description for
$U(2)$-types $\typeU=(\mu_{d,k})$; faithfulness is equivalent to
\[
\gcd\{d+2k:\mu_{d,k}>0\}=1
\]
(Proposition~\ref{prop:U2_types}). For $n\ge2$, the two dictionaries have different topological behavior: the $SU(2)$ dictionary is
compact, whereas the $U(2)$ dictionary is dense and not closed. Both are Haar-null, with Hausdorff dimensions $N^2-N+1$ and $N^2-N+2$,
respectively (Propositions~\ref{prop:spectral_consequences} and~\ref{prop:U2_spectral}). For $n=1$, they are $SU(2)$ and $U(2)$. We also
determine the unique minimal and maximal embedding strata for $N\ge4$ (Proposition~\ref{prop:extremal_strata}) and show that each
elementary gate has only finitely many $SU(2)$ spectral descriptors.

\paragraph{Distances to compact sectors.}
The Frobenius and operator-norm distances to the $SU(2)$ dictionary and to each fixed $U(2)$ sector reduce to eigenvalue matching
through the Hoffman--Wielandt and Bhatia--Davis inequalities \cite{HoffmanWielandt1953,BhatiaDavis1984} (Theorem~\ref{thm:matching}).
For the two-level dictionaries, Corollary~\ref{cor:closed_forms} evaluates these distances in closed form. The distance to the full
$U(2)$ dictionary is zero for every target because that dictionary is dense; for targets outside it, the infimum is not attained. The
spectral characterization of this countable union does not by itself provide a terminating procedure for deciding exact membership
(Remark~\ref{rem:U2_dense}).

\paragraph{Metric and operational consequences.}
An embedded $SU(2)$ subgroup of type $\typeSU$, with the induced Hilbert--Schmidt metric, is a round sphere of radius
$\sqrt{c_{\typeSU}}$, where
\[
c_{\typeSU}=\sum_dm_d\binom{d+2}{3}.
\]
Its geodesics from the identity remain minimizing in $U(N)$ exactly up to angle $\pi/d_{\max}(\typeSU)$
(Propositions~\ref{prop:dynkin} and~\ref{prop:ambient_minimizing}). Proposition~\ref{prop:coherence} gives a spectral bound on coherence
production that is sharp to first order near the identity. The isotropic distance also bounds from below penalized distances on $SU(N)$
whenever their norms dominate the Hilbert--Schmidt norm (Remark~\ref{rem:nielsen}), connecting the results with the
geometric-complexity framework of Nielsen and collaborators \cite{NielsenDowlingGuDoherty2006,Nielsen2006,DowlingNielsen2008}; see
\cite[\S3]{AcevedoFalco2026} for the penalty-matrix formulation used there and for the analogous distinction between intrinsic and
implementation-dependent costs for quantum channels.

These intrinsic constraints coexist with a limitation: elementarity does not determine tensor locality or entangling power. For two
qubits, every faithful $SU(2)$ representation type admits both product and entangling gates (Proposition~\ref{prop:elem_vs_entangling}).
The native CNOT, CZ, SWAP and Toffoli matrices belong to the phase-inclusive dictionary, while no global phase multiple of them belongs
to the phase-free dictionary (Proposition~\ref{prop:phase_sensitivity}). A finite universal alphabet requires additional geometric
choices (Remark~\ref{rem:finite_extrinsic}); Proposition~\ref{prop:covariance} constructs coordinate two-level alphabets covariant under
finite monomial groups.

\paragraph{Classical ingredients and scope.}
The classification of embeddings by faithful representations and multiplicities is standard \cite{BrockerDieck1985,Hall2015,Knapp2002}.
The formula for $\ell^U$ rests on the length of the unitary group with respect to pseudoreflections.
That length problem is classical: Dieudonn\'e solved it for unitary groups over general fields
\cite{Dieudonne1955}, as described in \cite[Introduction]{DjokovicMalzan1982}, five years after Scherk had solved the orthogonal case with respect to
symmetries \cite[Theorem~1]{Scherk1950}; the underlying orthogonal generation statement is due to
Cartan (see \cite[pp.~20--22]{Dieudonne1948}). Over $\C$ the answer is simply
$\operatorname{rank}(U-I_N)$, which we rederive in Remark~\ref{rem:intrinsic_vs_coordinate}; see
\cite{DjokovicMalzan1982} for the indefinite groups $U(p,q)$. This provides the comparison with the
determinant-one factor count in Theorem~\ref{thm:lsu}.
The nonidentity elements of $\Gtl(n)$ are exactly the determinant-one complex bireflections:
unitaries with a two-dimensional moving space and age $1$.
The exact determinant-one two-level count follows from the classical reflection-length theorem of
Djokovi\'c--Malzan \cite[Theorem]{DjokovicMalzan1979}.
There a reflection is the involution $H_{\mathbf v}=I_N-2\mathbf v\mathbf v^\dagger$,
$\|\mathbf v\|=1$, of determinant $-1$; such reflections generate
$\{A\in U(N):\det A=\pm1\}$, and we write $L_{\mathrm{ref}}$ for length with respect to them.
Their invariant $k_{\mathrm{DM}}(A)$ is the sum of the eigenangles in $[0,2\pi)$ divided by $\pi$, so that
$k_{\mathrm{DM}}(A)=2\operatorname{age}(A)$, and their Theorem states
\[
L_{\mathrm{ref}}(A)=\max\{k_{\mathrm{DM}}(A),k_{\mathrm{DM}}(\bar A)\}
\qquad\text{for } A\in U(N),\ \det A=\pm1 .
\]
Since $\spec(\bar A)=\spec(A^{-1})$, for $U\in SU(N)$ one has $k_{\mathrm{DM}}(U)=2\nu(U)$ and
$k_{\mathrm{DM}}(\bar U)=k_{\mathrm{DM}}(U^{-1})=2\big(r-\nu(U)\big)$, so that
\[
L_{\mathrm{ref}}(U)=2\max\{\nu(U),\,r-\nu(U)\}.
\]
Moreover $L_{\mathrm{ref}}(U)=2\ell^{SU}(U)$: each nonidentity two-level $SU(2)$ factor is a
product of exactly two reflections, and conversely a reflection factorization of $U\in SU(N)$ has
even length, so grouping consecutive pairs yields factors of determinant $1$ and rank at most $2$.
Theorem~\ref{thm:lsu}(ii) follows.
The bound $L_{\mathrm{ref}}\le2N-1$ and the conjecture that it is sharp are due to Radjavi
\cite{Radjavi1975}, and \cite{DjokovicMalzan1979} settles the conjecture.
Since $k_{\mathrm{DM}}(A)+k_{\mathrm{DM}}(\bar A)=2\operatorname{rank}(A-I_N)$, the value $2N-1$ is attained exactly when
$A-I_N$ is invertible and $\min\{k_{\mathrm{DM}}(A),k_{\mathrm{DM}}(\bar A)\}=1$, for instance when all eigenangles lie in $(0,\pi)$ with sum $\pi$;
such elements have determinant $-1$.
Within $SU(N)$ the maximum is therefore $\ell^{SU}=N-1$, attained at $e^{2\pi i/N}I_N$
(Remark~\ref{rem:lsu}).
The age originates in the McKay-correspondence literature, including Ito--Reid \cite{ItoReid1996Age}.
Our direct construction uses its subadditivity in the unitary formulation of
Guralnick--Tiep \cite[Theorem~2.3(i)]{GuralnickTiep2010Age}.
The active-support criterion provides the complete family of factors producing a maximal rank reduction. The phase-free spectral observation underlying Theorem~\ref{thm:spectral} was
pointed out by the referee. Coordinate Givens synthesis, error accumulation and Solovay--Kitaev compilation
\cite[\S\S4.5.1--4.5.4]{NielsenChuang2010}, \cite{Kitaev2002,DawsonNielsen2006}, and total geodesy for bi-invariant metrics
\cite[\S21]{Milnor1963}, are classical ingredients. Appendix~\ref{app:compilation} records the compilation results, including
approximation of the diagonal remainder (Corollary~\ref{cor:pure_word}). The Givens count is an upper bound for a particular coordinate
algorithm, whereas $\ell^U$ minimizes over arbitrary supporting planes; their comparison (Remark~\ref{rem:intrinsic_vs_coordinate}) does
not establish a circuit-complexity advantage.

\paragraph{Organization.}
Section~\ref{sec:strata} fixes notation and describes representation types and embedding strata. Section~\ref{sec:spectral} establishes
the spectral characterizations and their consequences. Section~\ref{sec:two_level} gives the Grassmannian model of the two-level stratum.
Section~\ref{sec:distance} treats distances, intrinsic factorization lengths and finite alphabets. Section~\ref{sec:variational} develops
the metric and coherence results. Section~\ref{sec:conclusion} states the open problems, and Appendix~\ref{app:compilation} collects the
classical compilation material.

\section{Embeddings, representations and strata}
\label{sec:strata}

\subsection{Notation and conventions}
\emph{Spaces and groups.} $\N=\{1,2,\dots\}$, $\N_0=\{0,1,\dots\}$. Fix $n\in\N$ and $N:=2^n$; we suppress $n$ from the notation of the
dictionaries when $N$ is fixed. $\C^N\cong(\C^2)^{\otimes n}$ carries the standard inner product and the computational basis
$\{\mathbf e_1,\dots,\mathbf e_N\}$, identified with the binary labels $\ket{b_1\cdots b_n}$ in lexicographic order; $\ket{\pm}:=(\ket0\pm\ket1)/\sqrt2$, and
$X,Y,Z$ denote the Pauli matrices $\sigma_x,\sigma_y,\sigma_z$. $A^\dagger$ is the adjoint; $I_r$ is the $r\times r$ identity and $I_V$
the identity of a space $V$ (we write $I$ only when the space is clear). Bold lowercase symbols denote column vectors
($\mathbf e_i$, $\mathbf q_i$, $\bm{\psi}$, \dots) and indexed families, such as the multiplicity families $\typeSU$, $\typeU$,
$\typeTwo$ of Section~\ref{sec:strata} or a family $\mathbf f$ of frames; their scalar components ($m_d$, $\mu_{d,k}$, \dots) are written in
ordinary type. Matrices are denoted by capital letters, and kets keep their usual labels. $S^1=\{z\in\C:|z|=1\}$, $\mu_r\subset S^1$ is the group of $r$-th
roots of unity, $\mathfrak S_N$ is the symmetric group, $\mathfrak u(N)=\{X:X^\dagger=-X\}$ and
$\mathfrak{su}(N)=\{X\in\mathfrak u(N):\Tr X=0\}$. $\Gr_2(\C^N)$ is the set of complex $2$-planes $W\subset\C^N$; $U(W)$, $SU(W)$ are the
(special) unitary groups of $W$, $P_W$ is the orthogonal projector onto $W$, and $\Fix(H)=\{\mathbf v:h\mathbf v=\mathbf v\ \forall h\in H\}$.
$PSU(2)=SU(2)/\{\pm I_2\}$. $\Ad_Q(A)=QAQ^{-1}$, and $\Ad_Q\phi:=\Ad_Q\circ\phi$ for a homomorphism $\phi$. $T\le U(N)$ is the diagonal torus
and $T_0:=T\cap SU(N)$. Products are ordered, $\prod_{k=1}^KA_k=A_1\cdots A_K$, and the empty product is the identity. Dimensions of groups
and orbits are real dimensions; for subsets that are not manifolds (such as the dictionaries) ``dimension'' means Hausdorff dimension
$\dim_H$ with respect to a Riemannian metric. We use repeatedly that smooth maps are locally Lipschitz and therefore do not increase
Hausdorff dimension, that the Hausdorff dimension of a countable union is the supremum of the dimensions, and that an immersed manifold
of dimension $k$ has Hausdorff dimension $k$.

\emph{Norms and distances.} $\|\cdot\|$ is the operator norm and $\|\cdot\|_F$ the Frobenius norm. On $\mathfrak u(N)$ we use the
$\Ad$-invariant product $\langle X,Y\rangle_{\mathrm{hs}}:=\tfrac12\Tr(X^\dagger Y)$, so that $\|X\|_{\mathrm{hs}}=\|X\|_F/\sqrt2$; it
induces a bi-invariant Riemannian metric on $U(N)$. For $\mathcal C\subset U(N)$ we write $\dist_F(U,\mathcal C):=\inf_{V\in\mathcal C}\|U-V\|_F$
and $\dist_{\mathrm{op}}(U,\mathcal C):=\inf_{V\in\mathcal C}\|U-V\|$ (extrinsic distances), whereas for a Lie subgroup $H\le U(N)$ with the
induced metric $d_H$ denotes the geodesic distance. The two must not be confused: for $V=\diag(-1,-1,1,\dots,1)$,
$\|V-I_N\|_F=2\sqrt2$ but $d_{U(N)}(I_N,V)=\pi$.

\emph{Spectra and logarithms.} $\spec(U)$ is the multiset of the $N$ eigenvalues of $U$, counted with multiplicity. Braces in spectral
formulas denote indexed multisets (repetitions are retained), $\biguplus$ denotes multiset union, $m\cdot\lambda$ denotes $m$ copies of a
multiset $\lambda$, and $\lambda^{(k)}$ denotes $k$ copies of $\lambda\in S^1$; ``$\spec(U)\subset A$'' refers to the values only. An
\emph{eigenangle} of a unitary matrix $U$ is a real number $\theta$ such that $e^{i\theta}$ is an eigenvalue of $U$; eigenangles are
defined modulo $2\pi$. We fix the branch $\Arg:S^1\to(-\pi,\pi]$ and choose the \emph{principal eigenangles} $\theta_k=\Arg\lambda_k\in(-\pi,\pi]$,
counting multiplicities. The principal logarithm of $U$ is $X_U:=Q\diag(i\theta_1,\dots,i\theta_N)Q^{-1}$ for any unitary diagonalization
$U=Q\diag(\lambda_1,\dots,\lambda_N)Q^{-1}$. It is skew-Hermitian and well defined (a spectral function of $U$), and it has minimal
Hilbert--Schmidt norm among all skew-Hermitian logarithms of $U$ (proof of Proposition~\ref{prop:ambient_minimizing}); when
$-1\in\spec(U)$ minimal logarithms need not be unique.

\emph{Positive eigenangles and age.}
For eigenvalues $\lambda\ne1$ we also use $\vartheta\in(0,2\pi)$ with $e^{i\vartheta}=\lambda$;
this is $\Arg\lambda$ when $\Arg\lambda>0$ and $\Arg\lambda+2\pi$ when $\Arg\lambda<0$.
The active eigenangles are listed increasingly when used in Theorem~\ref{thm:lsu}.
For $X\in U(N)$, let $\operatorname{age}(X)$ be the sum of its eigenangles in $[0,2\pi)$ divided by $2\pi$,
with eigenvalue $1$ contributing zero \cite[Definition~2.1]{GuralnickTiep2010Age}.
For $U\in SU(N)$ write $\nu(U):=\operatorname{age}(U)$; it is an integer because $\det U=1$.
With $r:=\operatorname{rank}(U-I_N)$,
\[
\nu(U^{-1})=r-\nu(U),\qquad \nu(I_N)=0.
\]
The eigenangle criterion uses the positive branch on the active support and concerns the concrete matrix $U$,
not its class modulo global phase.

\emph{Phase conventions.} Gates are concrete matrices in $U(N)$. The labels ``phase-free'' and ``phase-inclusive'' below refer to the
domains $SU(2)$ and $U(2)$ of the embeddings; they do not mean working modulo global phase. Approximation up to a global phase is considered
only when explicitly stated, in Proposition~\ref{prop:covariance}(iii) and Corollary~\ref{cor:pure_word}.

Let $V_d\cong\Sym^d(\C^2)$, $d\in\N_0$, denote the irreducible unitary $SU(2)$-module of dimension $d+1$, with representation
$\rho_d:SU(2)\to U(V_d)$. For $a\in\R$ put $g_a:=\diag(e^{ia},e^{-ia})\in SU(2)$. In the weight basis of $V_d$,
\begin{equation}\label{eq:weights}
\rho_d(g_a)=\diag\big(e^{ida},e^{i(d-2)a},\dots,e^{-ida}\big),
\qquad
\chi_d(a):=\Tr\rho_d(g_a)=\sum_{j=0}^{d}e^{i(d-2j)a}=\sum_{j=0}^{d}\cos\big((d-2j)a\big).
\end{equation}
Every $g\in SU(2)$ is conjugate in $SU(2)$ to $g_a$ for a unique $a\in[0,\pi]$, namely $\cos a=\tfrac12\Tr g$.

\subsection{Embeddings and elementary dictionaries}

\begin{definition}\label{def:embedding}
For $G\in\{SU(2),U(2)\}$, $\Emb(G,U(N))$ is the set of injective continuous homomorphisms $\phi:G\to U(N)$.
\end{definition}

\begin{remark}\label{rem:embedding_automatic}
A continuous homomorphism of Lie groups is smooth \cite[Ch.~3]{Hall2015}, and $\ker d\phi_e=\Lie(\ker\phi)$; hence injectivity of
$\phi$ already implies injectivity of $d\phi_e$. Since $G$ is compact, $\phi(G)$ is a closed (hence embedded) Lie subgroup and
$\phi:G\to\phi(G)$ is an isomorphism of Lie groups. Thus $\Emb(G,U(N))$ is exactly the set of faithful $N$-dimensional unitary
representations of $G$, and the conjugation action $(\Ad_W\phi)(g)=W\phi(g)W^{-1}$ of $U(N)$ corresponds to unitary equivalence.
\end{remark}

The dictionaries $\Gel(n)$ and $\GelU(n)$ are defined in the Introduction. Since $SU(2)$ is perfect, $\det\circ\phi$ is trivial for
every homomorphism $\phi:SU(2)\to U(N)$; hence
\begin{equation}\label{eq:Gel_in_SU}
\Gel(n)\subset SU(N).
\end{equation}
Once the tensor structure is fixed, the one-qubit placements
$w_j(V)=I_2^{\otimes(j-1)}\otimes V\otimes I_2^{\otimes(n-j)}$, $V\in U(2)$, are elements of $\Emb(U(2),U(N))$; the group they generate
is the \emph{local group}
\[
\Loc:=\{U_1\otimes\cdots\otimes U_n:\ U_j\in U(2)\}\cong U(2)^n/Z_n,\qquad
Z_n:=\Big\{(z_1I_2,\dots,z_nI_2):z_j\in S^1,\ \prod_jz_j=1\Big\},
\]
which is a proper subgroup of $U(N)$ for $n\ge2$.

\subsection{Homogeneous strata}

Let $\mF(N)$ be the set of finitely supported families $\typeSU=(m_d)_{d\ge0}$ of non-negative integers with $\sum_d m_d(d+1)=N$ and
$m_d>0$ for some odd $d$. We say that $\phi\in\Emb(SU(2),U(N))$ has \emph{type} $\typeSU$ if $\C^N\cong\bigoplus_d\C^{m_d}\otimes V_d$ as
$SU(2)$-modules. For $\typeSU\in\mF(N)$ we put
\begin{equation}\label{eq:cm_dmax}
c_{\typeSU}:=\sum_dm_d\binom{d+2}{3},\qquad d_{\max}(\typeSU):=\max\{d:m_d>0\}\ge1,
\end{equation}
and write $d_{\max}$ when $\typeSU$ is fixed. For $\phi\in\Emb(SU(2),U(N))$ we write $K_\phi:=\phi(SU(2))$ for the corresponding embedded
subgroup.

\begin{theorem}[Homogeneous stratification]\label{thm:strata}
For every $\typeSU\in\mF(N)$ the set $\mathcal O_{\typeSU}$ of embeddings of type $\typeSU$ is a single $U(N)$-orbit, and
\[
\Emb(SU(2),U(N))=\bigsqcup_{\typeSU\in\mF(N)}\mathcal O_{\typeSU},
\qquad
\mathcal O_{\typeSU}\cong U(N)\Big/\prod_{d:\,m_d>0}U(m_d),
\qquad
\dim\mathcal O_{\typeSU}=N^2-\sum_d m_d^2 .
\]
The number of strata is $|\mF(N)|=p(N)-q(N)$, where $p(N)$ is the number of partitions of $N$ and $q(N)$ the number of partitions of
$N$ into odd parts (equivalently, into distinct parts).
\end{theorem}

\begin{proof}
By complete reducibility and uniqueness of isotypic multiplicities for representations of compact groups
\cite[Ch.~II]{BrockerDieck1985}, two homomorphisms $SU(2)\to U(N)$ are $U(N)$-conjugate if and only if they have the same type. Since
$-I_2\in SU(2)$ acts on $V_d$ by $(-1)^d$ and $\{\pm I_2\}$ is the only non-trivial proper normal subgroup of $SU(2)$, a homomorphism of
type $\typeSU$ is faithful if and only if $m_d>0$ for some odd $d$. The stabilizer of $\phi$ is the centralizer of $\phi(SU(2))$, which by
Schur's lemma is $\{\bigoplus_d A_d\otimes I_{V_d}: A_d\in U(m_d)\}\cong\prod_dU(m_d)$; the orbit is the homogeneous manifold
$U(N)/\prod_dU(m_d)$ and its dimension is $N^2-\sum_dm_d^2$. Finally, $\typeSU\mapsto$ (the partition of $N$ with $m_d$ parts equal to $d+1$)
is a bijection from the set of all types onto the partitions of $N$, and faithfulness means that some part is even; the partitions with
only odd parts are counted by $q(N)$, which equals the number of partitions into distinct parts by Euler's theorem \cite[Ch.~1]{Andrews1998}.
\end{proof}

\begin{proposition}[Extremal strata]\label{prop:extremal_strata}
Let $N\ge4$. Among the strata of $\Emb(SU(2),U(N))$:
\begin{enumerate}
\item[(i)] the two-level stratum, of type $V_1\oplus V_0^{\oplus(N-2)}$, is the unique stratum of minimal dimension, $4N-5$;
\item[(ii)] the irreducible stratum, of type $V_{N-1}$, is the unique stratum of maximal dimension, $N^2-1$.
\end{enumerate}
\end{proposition}

\begin{proof}
(i) Put $s:=\sum_{d\ge1}m_d\ge1$. Then $m_0=N-\sum_{d\ge1}m_d(d+1)\le N-2s$ and $\sum_{d\ge1}m_d^2\le s^2$, so
$\sum_dm_d^2\le f(s):=(N-2s)^2+s^2$ for $1\le s\le N/2$. Since $f$ is convex, $\max f=\max\{f(1),f(N/2)\}$, and
$f(1)-f(N/2)=\tfrac14(3N-10)(N-2)>0$ for $N\ge4$. Equality $\sum_dm_d^2=f(1)$ forces $s=1$, $m_0=N-2$, i.e.\ $m_1=1$. (ii) $\sum_dm_d^2\ge1$, with
equality if and only if a single irreducible occurs once, i.e.\ $m_{N-1}=1$; this type is faithful because $N-1$ is odd.
\end{proof}

\begin{example}[$N=4$ and $N=8$]\label{ex:smallN}
For $N=4$ there are exactly $p(4)-q(4)=5-2=3$ strata:
\[
\begin{array}{lllll}
\toprule
\text{type} & \text{partition} & \text{stabilizer} & \dim\mathcal O_{\typeSU} & \text{example}\\
\midrule
V_1\oplus V_0^{\oplus2} & 2+1+1 & U(1)\times U(2) & 11 & \text{two-level (Section~\ref{sec:two_level})}\\
V_1^{\oplus2} & 2+2 & U(2) & 12 & \text{tensor placements } w_1|_{SU(2)},\,w_2|_{SU(2)}\\
V_3 & 4 & U(1) & 15 & \text{spin-}\tfrac32\\
\bottomrule
\end{array}
\]
For $N=8$ there are $p(8)-q(8)=22-6=16$ strata, of dimensions between $27$ (two-level) and $63$ (irreducible); for $N=16$ there are
$231-32=199$. In general the tensor placements $w_j|_{SU(2)}$ have type $V_1^{\oplus N/2}$ and orbit dimension $3N^2/4$.
\end{example}

\subsection{Faithful types of \texorpdfstring{$U(2)$}{U(2)}}

The irreducible unitary $U(2)$-modules are $V_{d,k}:=\Sym^d(\C^2)\otimes{\det}^k$, $d\in\N_0$, $k\in\Z$ \cite[Ch.~II]{BrockerDieck1985},
with representation $\rho_{d,k}$; in the monomial basis $x^{d-j}y^j$,
\begin{equation}\label{eq:weights_U2}
\rho_{d,k}\big(\diag(z_1,z_2)\big)=\diag\big(z_1^{d-j+k}z_2^{j+k}\big)_{j=0}^{d},
\qquad
\rho_{d,k}(zI_2)=z^{d+2k}I_{d+1}.
\end{equation}
Let $\mFU(N)$ be the set of finitely supported families $\typeU=(\mu_{d,k})_{(d,k)\in\N_0\times\Z}$ of non-negative integers with
$\sum_{d,k}\mu_{d,k}(d+1)=N$ and
\[
\gcd\{d+2k:\ \mu_{d,k}>0\}=1 .
\]
A homomorphism $\Phi:U(2)\to U(N)$ has \emph{type} $\typeU$ if $\C^N\cong\bigoplus_{d,k}\C^{\mu_{d,k}}\otimes V_{d,k}$.

\begin{proposition}[Faithful $U(2)$ types]\label{prop:U2_types}
Let $\Phi:U(2)\to U(N)$ be a continuous homomorphism of type $\typeU$. Then:
\begin{enumerate}
\item[(i)] $\Phi\in\Emb(U(2),U(N))$ if and only if $\typeU\in\mFU(N)$;
\item[(ii)] for $\typeU\in\mFU(N)$ the embeddings of type $\typeU$ form a single $U(N)$-orbit
$\mathcal O^U_{\typeU}\cong U(N)/\prod_{(d,k):\,\mu_{d,k}>0}U(\mu_{d,k})$; the set $\mFU(N)$ is countable, it has exactly two elements
for $N=2$, namely $\mu_{1,0}=1$ and $\mu_{1,-1}=1$, and it is infinite for $N\ge3$;
\item[(iii)] if $\typeU\in\mFU(N)$, then $\Phi|_{SU(2)}$ has type $\typeSU$ given by $m_d=\sum_k\mu_{d,k}$, and $\typeSU\in\mF(N)$.
\end{enumerate}
\end{proposition}

\begin{proof}
(i) Let $H\trianglelefteq U(2)$ be a normal subgroup. Then $H\cap SU(2)$ is normal in $SU(2)$, hence equal to $\{I_2\}$, $\{\pm I_2\}$ or
$SU(2)$. In the first two cases $H$ is central: for $h\in H$, the map $g\mapsto ghg^{-1}h^{-1}$ sends the connected group $U(2)$ into the
discrete set $H\cap SU(2)$ and equals $I_2$ at $g=I_2$. Apply this to $H=\ker\Phi$. By \eqref{eq:weights_U2}, $\ker\Phi$ meets the centre
$\{zI_2\}$ in $\{z:z^{e}=1\ \forall e\in E\}$, $E:=\{d+2k:\mu_{d,k}>0\}$, which is trivial if and only if $\gcd E=1$ (if $E=\{0\}$, it
is all of $S^1$). If $\gcd E=1$, some $d+2k$ is odd, so some $d$ is odd and $SU(2)\not\subset\ker\Phi$; hence $\ker\Phi$ is central and
therefore trivial. Conversely, injectivity on the centre forces $\gcd E=1$. (ii) As in the proof of Theorem~\ref{thm:strata}, by complete
reducibility and Schur's lemma. $\mFU(N)$ is countable because types are finitely supported families on $\N_0\times\Z$. For $N=2$, a
type is either a single $V_{1,k}$, faithful if and only if $|1+2k|=1$, i.e.\ $k\in\{0,-1\}$, or a sum of two characters $V_{0,k}\oplus
V_{0,k'}$, whose weights $2k,2k'$ are even, so that it is never faithful. For $N\ge3$ the types $\mu_{1,k}=1$, $\mu_{0,-k}=N-2$, $k\in\Z$,
are pairwise distinct and faithful, since $E=\{1+2k,-2k\}$ has $\gcd 1$. (iii) $V_{d,k}|_{SU(2)}=V_d$, and $\Phi|_{SU(2)}$ is injective, so $\typeSU\in\mF(N)$ by
Theorem~\ref{thm:strata}.
\end{proof}

For example, the two-level embeddings $\Phi_{W,f}$ of Section~\ref{sec:two_level} have type $\mu_{1,0}=1$, $\mu_{0,0}=N-2$, and the
tensor placements $w_j$ have type $\mu_{1,0}=N/2$; the $U(2)$-module $V_{N-1,0}$ is irreducible but not faithful for $N\ge4$
(its centre acts through $z\mapsto z^{N-1}$), whereas $V_{N-1,1-N/2}$ is faithful.

\section{Spectral characterization of elementarity}
\label{sec:spectral}

For $\typeSU\in\mF(N)$ and $a\in[0,\pi]$ define the multiset
\begin{equation}\label{eq:Lambda}
\Lambda_{\typeSU}(a):=
\biguplus_{d\ge0}m_d\cdot
\{e^{i(d-2j)a}\}_{j=0}^{d},
\end{equation}
a multiset of $N$ points of $S^1$, and let
\[
\mathcal S_{\typeSU}:=\bigcup_{\phi\in\mathcal O_{\typeSU}}K_\phi\subset U(N)
\]
be the \emph{spectral sector} of type $\typeSU$, i.e.\ the set of gates that \emph{admit} type $\typeSU$. Thus $\Gel(n)=\bigcup_{\typeSU\in\mF(N)}\mathcal S_{\typeSU}$.
Unlike the strata $\mathcal O_{\typeSU}$, which are disjoint subsets of $\Emb(SU(2),U(N))$, the sectors are subsets of $U(N)$ and overlap (for
instance $I_N$ lies in all of them), so a gate may admit several types.

\begin{theorem}[Spectral characterization, phase-free]\label{thm:spectral}
For every $\typeSU\in\mF(N)$,
\[
\mathcal S_{\typeSU}=\big\{U\in U(N):\ \spec(U)=\Lambda_{\typeSU}(a)\ \text{for some } a\in[0,\pi]\big\}.
\]
Consequently, $U\in\Gel(n)$ if and only if $\spec(U)=\Lambda_{\typeSU}(a)$ for some $\typeSU\in\mF(N)$ and some $a\in[0,\pi]$.
\end{theorem}

\begin{proof}
Fix a reference embedding $\phi_{\typeSU}:=\bigoplus_d I_{m_d}\otimes\rho_d$ of type $\typeSU$, written in an orthonormal basis adapted to the
weight bases; by \eqref{eq:weights}, $\phi_{\typeSU}(g_a)$ is diagonal with $\spec\phi_{\typeSU}(g_a)=\Lambda_{\typeSU}(a)$.

($\subseteq$) Let $U=\phi(g)$ with $\phi\in\mathcal O_{\typeSU}$, $g\in SU(2)$. Then $\phi=\Ad_W\phi_{\typeSU}$ for some $W\in U(N)$, and
$g=hg_ah^{-1}$ for some $h\in SU(2)$ and $a\in[0,\pi]$. Hence $U=(W\phi_{\typeSU}(h))\,\phi_{\typeSU}(g_a)\,(W\phi_{\typeSU}(h))^{-1}$ and
$\spec(U)=\Lambda_{\typeSU}(a)$.

($\supseteq$) If $\spec(U)=\Lambda_{\typeSU}(a)$, then $U$ and $\phi_{\typeSU}(g_a)$ are normal with the same eigenvalue multiset, hence unitarily
similar: $U=Q\phi_{\typeSU}(g_a)Q^{-1}$ for some $Q\in U(N)$. Then $\phi:=\Ad_Q\phi_{\typeSU}\in\mathcal O_{\typeSU}$ and $U=\phi(g_a)\in\mathcal S_{\typeSU}$.
\end{proof}

For $\typeU\in\mFU(N)$ and $z_1,z_2\in S^1$ define analogously
\begin{equation}\label{eq:LambdaU}
\Lambda^U_{\typeU}(z_1,z_2):=\biguplus_{(d,k)}\ \mu_{d,k}\cdot\big\{z_1^{d-j+k}z_2^{j+k}\big\}_{j=0}^{d},
\qquad
\mathcal S^U_{\typeU}:=\bigcup_{\Phi\in\mathcal O^U_{\typeU}}\Phi\big(U(2)\big),
\end{equation}
so that $\GelU(n)=\bigcup_{\typeU\in\mFU(N)}\mathcal S^U_{\typeU}$; note $\Lambda^U_{\typeU}(z_1,z_2)=\Lambda^U_{\typeU}(z_2,z_1)$ (substitute $j\mapsto d-j$).

\begin{theorem}[Spectral characterization, phase-inclusive]\label{thm:spectral_U2}
For every $\typeU\in\mFU(N)$,
\[
\mathcal S^U_{\typeU}=\big\{U\in U(N):\ \spec(U)=\Lambda^U_{\typeU}(z_1,z_2)\ \text{for some } z_1,z_2\in S^1\big\}.
\]
Consequently, $U\in\GelU(n)$ if and only if there are $\typeU\in\mFU(N)$, i.e.\ data $(d_\alpha,k_\alpha,\mu_\alpha)_\alpha$ with
$\sum_\alpha\mu_\alpha(d_\alpha+1)=N$ and $\gcd_\alpha(d_\alpha+2k_\alpha)=1$, and $z_1,z_2\in S^1$ such that
$\spec(U)=\biguplus_\alpha\mu_\alpha\cdot\{z_1^{d_\alpha-j+k_\alpha}z_2^{j+k_\alpha}\}_{j=0}^{d_\alpha}$.
\end{theorem}

\begin{proof}
Identical to that of Theorem~\ref{thm:spectral}, with the reference embedding $\Phi_{\typeU}:=\bigoplus_{d,k}I_{\mu_{d,k}}\otimes\rho_{d,k}$
(an embedding by Proposition~\ref{prop:U2_types}(i)), formula \eqref{eq:weights_U2}, and the fact that every element of $U(2)$ is
$U(2)$-conjugate to a diagonal matrix $\diag(z_1,z_2)$.
\end{proof}

For instance, $\spec(U)=\{z_1,z_2,1^{(N-2)}\}$ for the two-level type, in accordance with Proposition~\ref{prop:spectral_consequences}(ii)
below; and the spectrum $\{1^{(N-1)},-1\}$ of CNOT, CZ, SWAP and Toffoli is $\Lambda^U_{\typeU}(-1,1)$ for that type.

\begin{proposition}[Consequences, phase-free]\label{prop:spectral_consequences}
The following statements hold.
\begin{enumerate}
\item[(i)] Every $U\in\Gel(n)$ satisfies $\det U=1$ and $\overline{\spec(U)}=\spec(U)$ (as multisets).
\item[(ii)] For the two-level type $\typeTwo$, given by $t_1=1$, $t_0=N-2$ and $t_d=0$ for $d\ge2$,
\[
\Gtl(n)=\mathcal S_{\typeTwo}=\{U\in SU(N):\operatorname{rank}(U-I)\le2\},
\qquad
\GtlU(n)=\{U\in U(N):\operatorname{rank}(U-I)\le2\},
\]
where $\Gtl(n)$, $\GtlU(n)$ are the two-level dictionaries of Section~\ref{sec:two_level}.
\item[(iii)] Each $\mathcal S_{\typeSU}$ is compact, and $\Gel(n)$ is a compact subset of $SU(N)$ with
\[
\dim_H\Gel(n)=N^2-N+1 .
\]
In particular, for $n\ge2$, $\Gel(n)$ has Haar measure zero in $SU(N)$, although it generates $SU(N)$
(Corollary~\ref{cor:universality_elem}). For $N=4$ the dimension is $13$, whereas $\dim\mathcal O_{\typeSU}=15$ for the irreducible type.
\item[(iv)] For every $U\in U(N)$ the set of \emph{spectral descriptors}
$\Sigma(U):=\{(\typeSU,a)\in\mF(N)\times[0,\pi]:\ \spec(U)=\Lambda_{\typeSU}(a)\}$ is finite; it is empty if and only if $U\notin\Gel(n)$.
\end{enumerate}
\end{proposition}

\begin{proof}
(i) $\det U=1$ is \eqref{eq:Gel_in_SU}; $\Lambda_{\typeSU}(a)$ is invariant under $j\mapsto d-j$, which conjugates each eigenvalue.

(ii) By Theorem~\ref{thm:spectral}, $U\in\mathcal S_{\typeTwo}$ if and only if $\spec(U)=\{e^{ia},e^{-ia},1^{(N-2)}\}$. If $U\in SU(N)$ and
$\operatorname{rank}(U-I)\le2$, then $1$ has multiplicity at least $N-2$ and the remaining two eigenvalues have product $\det U=1$, so
$\spec(U)$ has this form; the converse is clear. An element of $SU(N)[W]$ (Definition~\ref{def:two_level_subgroup}) has this
spectrum; conversely, if $\operatorname{rank}(U-I)\le2$, choose a $2$-plane $W\supseteq\operatorname{Ran}(U-I)=\ker(U-I)^\perp$;
since $U$ is normal, $W$ is invariant, $U|_{W^\perp}=I$, and $U\in SU(N)[W]$ if $\det U=1$ (resp.\ $U\in U(N)[W]$ in general).

(iii) $\mathcal S_{\typeSU}$ is the image of the compact set $U(N)\times[0,\pi]$ under the continuous map
$(Q,a)\mapsto Q\phi_{\typeSU}(g_a)Q^{-1}$. Since $\phi_{\typeSU}(g_a)$ is diagonal, this map factors through the smooth manifold with boundary
$U(N)/T\times[0,\pi]$ of dimension $N^2-N+1$; hence $\dim_H\mathcal S_{\typeSU}\le N^2-N+1$ and, $\mF(N)$ being finite,
$\dim_H\Gel(n)\le N^2-N+1$, which is less than $\dim SU(N)=N^2-1$ as soon as $N\ge3$. For the irreducible type $\typeSU^{\mathrm{irr}}$
($m_{N-1}=1$) and $a\in(0,\pi/(N-1))$ the eigenvalues $e^{i(N-1-2j)a}$ are distinct and determine $a$, and $U(N)/T\times T_{\mathrm{reg}}\to U(N)$,
$(Q,D)\mapsto QDQ^{-1}$, is a local diffeomorphism on the regular set; hence the map is an injective immersion of
$U(N)/T\times(0,\pi/(N-1))$, and $\dim_H\Gel(n)\ge N^2-N+1$.

(iv) $\mF(N)$ is finite. For fixed $\typeSU$, the condition $\spec(U)=\Lambda_{\typeSU}(a)$ implies $\sum_dm_d\chi_d(a)=\Tr U$. The left-hand side
is a real cosine polynomial in $a$ whose top frequency $d_{\max}(\typeSU)\ge1$ has coefficient $2m_{d_{\max}}>0$ (the terms $j=0$ and
$j=d_{\max}$ of $\chi_{d_{\max}}$); hence it is non-constant and takes the value $\Tr U$ at finitely many $a\in[0,\pi]$. The last claim is
Theorem~\ref{thm:spectral}.
\end{proof}

\begin{remark}[Descriptor selection]\label{rem:descriptor_selection}
Proposition~\ref{prop:spectral_consequences}(iv) settles the \emph{finite} part of the descriptor-selection problem: the admissible
types and angles of a given $U\in\Gel$ are read off from its spectrum. The \emph{continuous} part does not disappear: for $(\typeSU,a)\in\Sigma(U)$
the set of embeddings $\phi\in\mathcal O_{\typeSU}$ with $U\in\phi(SU(2))$ is invariant under conjugation by the centralizer of $U$, and in
general this action has positive-dimensional orbits. For instance, let $U\in\Gtl(n)$ with $(\typeTwo,a)\in\Sigma(U)$, where $\typeTwo$ is the
two-level type. If $U\ne I_N$, the minimal support $\supp(U):=\Ran(U-I_N)$ is a $2$-plane $W$, hence unique; once the internal
reference element $g_a$ is fixed, the frames $f$ with $\phi_{W,f}(g_a)=U$ form a coset of the centralizer of $g_a$ in $U(2)$, which is a
maximal torus for $a\in(0,\pi)$ and all of $U(2)$ for $a=\pi$ (i.e.\ $U|_W=-I_W$). For $U=I_N$ every $W\in\Gr_2(\C^N)$ is admissible.
The spectral descriptors can be ranked by $\sqrt{c_{\typeSU}}\,a$, the length of a minimizing path from $I_N$ to $U$ inside a single embedded
subgroup of type $\typeSU$ (Proposition~\ref{prop:dynkin} and Remark~\ref{rem:nielsen}).
\end{remark}

\begin{proposition}[Consequences, phase-inclusive]\label{prop:U2_spectral}
The following statements hold.
\begin{enumerate}
\item[(i)] For every $U\in\GelU(n)$ there are $z_1,z_2\in S^1$ such that $\spec(U)$ is contained in the subgroup of $S^1$ generated by
$z_1,z_2$.
\item[(ii)] For $n\ge2$, $\GelU(n)$ has Haar measure zero in $U(N)$, and $\dim_H\GelU(n)=N^2-N+2$.
\item[(iii)] Fix the tensor structure, a set $C$ of control qubits and a disjoint ordered list of $k\ge2$ target qubits (suppressed
from the notation $\Lambda_C$), the remaining $s:=n-|C|-k$ qubits being spectators. For Haar-almost every $V\in U(2^k)$, the controlled
gate $\Lambda_C(V)$, which acts as $V\otimes I_{2^s}$ (after reordering the factors as targets, spectators) when all controls are $\ket1$
and as the identity otherwise, does not belong to $\GelU(n)$.
\item[(iv)] $\GelU(n)$ is dense in $U(N)$; hence $\dist_F(U,\GelU(n))=\dist_{\mathrm{op}}(U,\GelU(n))=0$ for every $U\in U(N)$, and for
$n\ge2$ the set $\GelU(n)$ is not closed.
\end{enumerate}
\end{proposition}

\begin{proof}
(i) is immediate from Theorem~\ref{thm:spectral_U2}. (ii) By Theorem~\ref{thm:spectral_U2}, $\mathcal S^U_{\typeU}$ is the image of
$U(N)/T\times(S^1)^2$, a manifold of dimension $N^2-N+2<N^2$ for $N\ge3$, under the smooth map
$(Q,z_1,z_2)\mapsto Q\Phi_{\typeU}(\diag(z_1,z_2))Q^{-1}$; $\mFU(N)$ is countable, so $\GelU(n)$ is Haar-null and
$\dim_H\GelU(n)\le N^2-N+2$. For the lower bound take $\mu_{N-1,k_0}=1$ with $k_0:=1-N/2$, which lies in $\mFU(N)$ because
$(N-1)+2k_0=1$. Writing $z_1=e^{i\alpha}$, $z_2=e^{i\beta}$, the eigenangles are
$\theta_j=(N-1-j+k_0)\alpha+(j+k_0)\beta$, $j=0,\dots,N-1$; the linear map $(\alpha,\beta)\mapsto(\theta_0,\theta_0-\theta_1)=
((N-1+k_0)\alpha+k_0\beta,\ \alpha-\beta)$ has determinant $-(N-1+2k_0)=-1$, and for $0<|\alpha-\beta|$ small the $\theta_j$ are distinct
modulo $2\pi$. As in the proof of Proposition~\ref{prop:spectral_consequences}(iii), the map is an immersion near such points, so
$\dim_H\GelU(n)\ge N^2-N+2$. (iii) $\spec\Lambda_C(V)$ consists of $2^s$ copies of $\spec(V)$ together with $N-2^{k+s}$ eigenvalues equal to
$1$; in particular $\spec(V)\subset\spec\Lambda_C(V)$, so $\Lambda_C(V)\in\GelU(n)$ forces $\spec(V)$ into a subgroup generated by two
elements; for each integer matrix $M\in\Z^{2^k\times2}$ the unitaries in $U(2^k)$ with eigenvalues $(z_1^{M_{i1}}z_2^{M_{i2}})_i$ form
the image of $U(2^k)/T\times(S^1)^2$, of dimension $4^k-2^k+2<4^k$ ($2^k\ge4$), and a countable union of such images is a null set. (iv) For $N=2$, $\GelU(1)=U(2)$. Let $N\ge3$ and $\varepsilon>0$; by conjugation it suffices to approximate a
diagonal $D=\diag(\lambda_1,\dots,\lambda_N)$. Choose $z_1,z_2\in S^1$ with $|z_i-\lambda_i|<\varepsilon$ such that $t:=z_1z_2$ is not a root
of unity; then $\{t^k:k\in\Z\}$ is dense in $S^1$, and we may choose $k_3,\dots,k_N\in\Z$ with $|t^{k_j}-\lambda_j|<\varepsilon$. The
homomorphism $\Phi(h):=h\oplus(\det h)^{k_3}\oplus\cdots\oplus(\det h)^{k_N}$ has a type containing $V_{1,0}$, hence lies in
$\Emb(U(2),U(N))$ by Proposition~\ref{prop:U2_types}(i), and $\Phi(\diag(z_1,z_2))=\diag(z_1,z_2,t^{k_3},\dots,t^{k_N})$ is within
$\varepsilon$ of $D$ in operator norm (and within $\sqrt N\varepsilon$ in Frobenius norm). For $n\ge2$, $\GelU(n)$ is a proper subset by
(ii), so a dense proper subset, which is not closed.
\end{proof}

\begin{proposition}[Phase sensitivity of the dictionaries]\label{prop:phase_sensitivity}
Let $N\ge4$ and let $U\in U(N)$ have $\spec(U)=\{1^{(N-1)},-1\}$ (e.g.\ the $4\times4$ matrices of CNOT, CZ or SWAP for $n=2$, the
$8\times8$ Toffoli matrix for $n=3$, or any multi-controlled $X$ or $Z$ with $n-1$ controls; tensor extensions by identities, such as
$\mathrm{CNOT}\otimes I_2\in U(8)$, have a different spectrum and are not covered). Then $U\in\GtlU(n)$, but $e^{i\theta}U\notin\Gel(n)$
for every $\theta\in\R$.
\end{proposition}

\begin{proof}
$\operatorname{rank}(U-I)=1$, so $U\in\GtlU(n)$ by Proposition~\ref{prop:spectral_consequences}(ii). Suppose $e^{i\theta}U\in\Gel(n)$ and
put $\lambda=e^{i\theta}$. By Proposition~\ref{prop:spectral_consequences}(i), $\{\lambda^{(N-1)},-\lambda\}=\{\bar\lambda^{(N-1)},-\bar\lambda\}$;
comparing the eigenvalue of multiplicity $N-1\ge3$ gives $\lambda=\bar\lambda$, i.e.\ $\lambda=\pm1$. But
$1=\det(\lambda U)=-\lambda^N=-1$ because $N$ is even, a contradiction.
\end{proof}

\begin{proposition}[Standard gates]\label{prop:standard_gates}
Fix the tensor structure, a set $C$ of control qubits and a target qubit $\tau\notin C$ (suppressed from the notation), and let
$\Lambda_C(V)$ apply $V\in U(2)$ to qubit $\tau$ when all qubits in $C$ are in $\ket1$. Then the following statements hold.
\begin{enumerate}
\item[(i)] $\Lambda_C\in\Emb(U(2),U(N))$, of type $\mu_{1,0}=2^{\,n-1-|C|}$, $\mu_{0,0}=N-2^{\,n-|C|}$. In particular
$\Lambda_C(U(2))$ is a two-level subgroup $U(N)[W]$ if and only if $|C|=n-1$.
\item[(ii)] Every transposition of two computational basis vectors of $\C^N$ lies in $\GtlU(n)$; examples are the matrices of CNOT
and SWAP for $n=2$, of Toffoli for $n=3$, and of $C^{n-1}X$ for every $n$. Tensor extensions require care: let $P\in U(M)$, $M=2^r$, be a
transposition of two basis vectors and $s\ge1$, so that $P\otimes I_{2^s}\in U(N)$ with $N=2^{r+s}$. Then
$\operatorname{rank}(P\otimes I_{2^s}-I_N)=2^s$. Hence for $s=1$, $\spec(P\otimes I_2)=\{1^{(N-2)},-1,-1\}$ and
$P\otimes I_2\in\Gtl(r+1)$, although it is not supported on a coordinate plane, whereas for $s\ge2$, $P\otimes I_{2^s}\notin\GtlU(r+s)$.
\end{enumerate}
\end{proposition}

\begin{proof}
(i) $\Lambda_C(V)$ acts as $I_{2^{\,n-1-|C|}}\otimes V$ on the $2^{\,n-|C|}$-dimensional subspace where the controls are $\ket1$ and as
the identity on its complement; $V\mapsto\Lambda_C(V)$ is an injective homomorphism, and the type follows. It is two-level iff
$\mu_{1,0}=1$; that the subgroup $\Lambda_C(U(2))$ is not two-level for $|C|<n-1$ does not prevent particular elements of it from being
two-level gates. (ii) A transposition of $\mathbf e_p,\mathbf e_q$ is the two-level unitary $X$ on $\mathrm{span}\{\mathbf e_p,\mathbf e_q\}$. If $P\in U(M)$ swaps
$\mathbf e_p,\mathbf e_q$, then $P\otimes I_{2^s}$ acts as $-I$ on the $2^s$-dimensional span of the vectors $(\mathbf e_p-\mathbf e_q)\otimes\mathbf f$, $\mathbf f$ in a basis of
$\C^{2^s}$, and as the identity on its complement.
\end{proof}

\paragraph{Elementarity versus entanglement.}
Elementarity is invariant under conjugation by $U(N)$.
Entangling power is invariant under local conjugation, by $\Loc$, but generally not under arbitrary unitary conjugation. The following proposition shows that the two notions are independent already for two qubits.

We call a two-qubit unitary \emph{entangling} if it maps some product state to an entangled state.

\begin{proposition}[$n=2$: every faithful representation type admits product and entangling gates]\label{prop:elem_vs_entangling}
Let $N=4$, $a\in(0,\pi/2)$ and $R_z(\theta)=\diag(e^{-i\theta/2},e^{i\theta/2})$. Each element in the following table lies in the spectral
sector $\mathcal S_{\typeSU}$ of the indicated representation type; the elements of the third column are products $A\otimes B$, and those of the
last column are entangling:
\[
\begin{array}{llll}
\toprule
\text{type } \typeSU & \text{spectrum} & \text{product element} & \text{entangling element}\\
\midrule
V_1\oplus V_0^{\oplus2} & \{e^{\pm ia},1,1\} & R_z(a)\otimes R_z(a) & \Lambda_{\{1\}}(e^{-iaY})\\
V_1^{\oplus2} & \{e^{\pm ia},e^{\pm ia}\} & e^{-iaZ}\otimes I & e^{-ia\,X\otimes X}\\
V_3 & \{e^{\pm3ia},e^{\pm ia}\} & e^{-2iaZ}\otimes e^{-iaZ} & e^{-ia(X\otimes X+2\,Z\otimes Z)}\\
\bottomrule
\end{array}
\]
\end{proposition}

\begin{proof}
The spectra are immediate ($R_z(a)\otimes R_z(a)=\diag(e^{-ia},1,1,e^{ia})$; $X\otimes X$ and $Z\otimes Z$ commute with joint
eigenvalues $(\pm1,\pm1)$ on the Bell basis), so each element lies in the indicated sector by Theorem~\ref{thm:spectral}. We exhibit
a product input with entangled output. For $\Lambda_{\{1\}}(V)=\ket0\!\bra0\otimes I+\ket1\!\bra1\otimes V$ with $V=e^{-iaY}$,
\[
\Lambda_{\{1\}}(V)\ket{+}\ket0=\tfrac1{\sqrt2}\big(\ket0\ket0+\ket1(\cos a\ket0+\sin a\ket1)\big),
\]
which is entangled because $\ket0$ and $\cos a\ket0+\sin a\ket1$ are linearly independent ($\sin a\neq0$). Next,
$e^{-iaX\otimes X}\ket{00}=\cos a\ket{00}-i\sin a\ket{11}$, which is entangled for $a\in(0,\pi/2)$. Finally, $X\otimes X$ and $Z\otimes Z$
commute and $Z\otimes Z\ket{00}=\ket{00}$, so $e^{-ia(X\otimes X+2Z\otimes Z)}\ket{00}=e^{-2ia}(\cos a\ket{00}-i\sin a\ket{11})$, again
entangled.
\end{proof}

\begin{remark}
Non-product (operator-Schmidt rank $>1$) and entangling are different properties: SWAP has operator-Schmidt rank $4$ but maps product
states to product states. For this reason Proposition~\ref{prop:elem_vs_entangling} is proved with explicit input states.
\end{remark}

\begin{remark}[Interpretation]\label{rem:interpretation}
Proposition~\ref{prop:elem_vs_entangling} reflects a structural fact: a conjugation-invariant class is a union of
$U(N)$-conjugacy classes, and conjugation by non-local unitaries does not preserve the product structure. ``Elementary'' in the present
sense therefore means \emph{Lie-theoretically simple} (a single one-parameter spin-multiplet motion, with spectrum of the form
\eqref{eq:Lambda} or \eqref{eq:LambdaU}), not \emph{local}. Locality and entangling power are extrinsic refinements, obtained by
restricting the conjugating group from $U(N)$ to the local group $\Loc$. Proposition~\ref{prop:phase_sensitivity} shows, conversely, that the
standard two- and three-qubit primitives CNOT, CZ, SWAP and Toffoli are elementary only in the phase-inclusive dictionary $\GelU$.
\end{remark}

\begin{remark}[Closed versus dense]\label{rem:U2_dense}
Assume $n\ge2$ (for $n=1$, $\GelU=U(2)$). Proposition~\ref{prop:U2_spectral}(iv) separates the two dictionaries. $\Gel$ is a finite union of compact sectors, so its distance to a
target is attained and computable by Theorem~\ref{thm:matching} below. Each phase-inclusive sector $\mathcal S^U_{\typeU}$ is compact as well,
but $\GelU$ is the union of infinitely many of them, with unbounded determinant twists, and this union is dense: the extrinsic distance
to $\GelU$ carries no information, and only the individual sectors, or families of sectors with bounded twists, have metric content.
Density is compatible with Haar-nullity and with $\dim_H\GelU=N^2-N+2$. Theorem~\ref{thm:spectral_U2} characterizes membership in
$\GelU$ completely, but since the twists are unbounded it does not by itself yield a terminating decision procedure; we do not address
this question here.
\end{remark}

\section{The two-level stratum and the Grassmannian}
\label{sec:two_level}

\begin{definition}\label{def:two_level_subgroup}
For a $2$-plane $W\in\Gr_2(\C^N)$ let
$U(N)[W]:=\{U\in U(N): U(W)=W,\ U|_{W^\perp}=I_{W^\perp}\}$ and $SU(N)[W]:=U(N)[W]\cap SU(N)$. For $U\in U(N)$ we call
$\supp(U):=\Ran(U-I_N)$ its \emph{minimal support}; since $U$ is normal, $U\in U(N)[W]$ if and only if $\supp(U)\subseteq W$, and we
then call $W$ a \emph{supporting plane} of $U$. The two-level dictionaries are
\[
\Gtl(n):=\bigcup_{W\in\Gr_2(\C^N)}SU(N)[W],\qquad \GtlU(n):=\bigcup_{W\in\Gr_2(\C^N)}U(N)[W].
\]
For an orthonormal frame $f:\C^2\to W$ put $\phi_{W,f}(S):=fSf^{-1}\oplus I_{W^\perp}$, $S\in SU(2)$, and
$\Phi_{W,f}(V):=fVf^{-1}\oplus I_{W^\perp}$, $V\in U(2)$. In the computational basis we write $W_{p,q}:=\mathrm{span}\{\mathbf e_p,\mathbf e_q\}$ and
$\phi_{p,q}:=\phi_{W_{p,q},(\mathbf e_p,\mathbf e_q)}$ (similarly $\Phi_{p,q}$).
\end{definition}

\begin{proposition}[Grassmannian model]\label{prop:grassmannian}
The following statements hold.
\begin{enumerate}
\item[(i)] $\phi_{W,f}\in\Emb(SU(2),U(N))$ has type $V_1\oplus V_0^{\oplus(N-2)}$, and $\phi_{W,f}(SU(2))=SU(N)[W]$ for every
frame $f$. If $f'=fu$ with $u\in U(2)$, then $\phi_{W,f'}=\phi_{W,f}\circ\Ad_{s}$ with $s:=(\det u)^{-1/2}u\in SU(2)$, well
defined up to $\pm I$; likewise $\Phi_{W,f'}=\Phi_{W,f}\circ\Ad_u$, and $\Ad_u$ depends only on the class of $u$ in
$U(2)/U(1)\cong PSU(2)$, where $U(1)=\{zI_2\}$ is the centre.
\item[(ii)] $W\mapsto SU(N)[W]$ is injective, since $\Fix(SU(N)[W])=W^\perp$.
\item[(iii)] The map $\pi:\mathcal O_{\typeTwo}\to\Gr_2(\C^N)$, $\phi_{W,f}\mapsto W$, is a $U(N)$-equivariant principal $PSU(2)$-bundle; it is the
homogeneous fibration
\[
U(N)/\big(U(1)\times U(N-2)\big)\longrightarrow U(N)/\big(U(2)\times U(N-2)\big)
\]
with fibre $U(2)/U(1)\cong PSU(2)\cong SO(3)$.
\end{enumerate}
In particular, the gauge group acting on two-level embedding data is $PSU(2)$, both for $SU(2)$-embeddings and for
$U(2)$-embeddings.
\end{proposition}

\begin{proof}
(i) The type is clear. Writing $u=(\det u)^{1/2}s$, the central factor cancels in $uSu^{-1}=sSs^{-1}$ and in $uVu^{-1}$. (ii) Every
element fixes $W^\perp$ pointwise, and $SU(W)$ has no common fixed vector in $W$. (iii) $U(N)$ acts transitively on frames, hence on
$\mathcal O_{\typeTwo}$; the stabilizer of $\phi_{W_{1,2},(\mathbf e_1,\mathbf e_2)}$ is its centralizer $U(1)\times U(N-2)$ (scalars on $W$, arbitrary on
$W^\perp$), and that of $W_{1,2}$ is $U(2)\times U(N-2)$. A homogeneous fibration $G/H\to G/K$ with $H\trianglelefteq K$ is a principal
$K/H$-bundle, which gives (iii).
\end{proof}

\begin{remark}[Descriptors and tensor locality]\label{rem:descriptor_pair}
In terms of a Stiefel frame $F\in\C^{N\times2}$, $F^\dagger F=I_2$, $\Ran F=W$ (so $P_W=FF^\dagger$), a two-level $SU(2)$ gate with
supporting plane $W$ is $U_{W,F}(S)=I_N+F(S-I_2)F^\dagger$, $S\in SU(2)$. The gauge-invariant datum of a gate \emph{with a chosen
supporting plane} is the class of the pair $(F,S)$ under $(F,S)\sim(Fu,u^{-1}Su)$, $u\in U(2)$; equivalently, the pair $(W,A)$ with
$A=FSF^\dagger|_W\in SU(W)$. For $U\in\Gtl\setminus\{I_N\}$ the supporting plane is unique, $W=\supp(U)$, so this datum is determined by
the gate; for $U=I_N$ (and, in $\GtlU$, for gates with $\operatorname{rank}(U-I_N)=1$) the plane is a genuine additional choice. For the
$U(2)$ version replace $S\in SU(2)$ by $V\in U(2)$ and $SU(W)$ by $U(W)$. The projector $P_W$ and the conjugacy class of $S$ alone do
not determine the gate: they determine only a supporting plane and the spectrum (e.g.\ $\diag(e^{ia},e^{-ia})$ and $e^{iaX}$ on the same
frame give different gates). The tensor placements $w_j|_{SU(2)}$ are of type $V_1^{\oplus N/2}$ and are therefore never
two-level for $n\ge2$ (Example~\ref{ex:smallN}).
\end{remark}

\section{Distances, intrinsic lengths and finite alphabets}
\label{sec:distance}

\subsection{Distances to the dictionaries}

Let $\mathscr X_N:=(S^1)^N/\mathfrak S_N$ be the space of $N$-point multisets on $S^1$, with the quotient topology. For
$\lambda=\{\lambda_1,\dots,\lambda_N\}$, $\mu=\{\mu_1,\dots,\mu_N\}\in\mathscr X_N$ define the optimal matching distances
\[
\delta_2(\lambda,\mu):=\min_{\pi\in\mathfrak S_N}\Big(\sum_{i=1}^N|\lambda_i-\mu_{\pi(i)}|^2\Big)^{1/2},
\qquad
\delta_\infty(\lambda,\mu):=\min_{\pi\in\mathfrak S_N}\max_i|\lambda_i-\mu_{\pi(i)}|,
\]
both of which metrize the quotient topology. For a \emph{non-empty compact} family $\mathscr M\subset\mathscr X_N$ we call
$\mathcal C_{\mathscr M}:=\{V\in U(N):\spec(V)\in\mathscr M\}$ \emph{spectrally defined}.

\begin{theorem}[Distances as eigenvalue matching]\label{thm:matching}
Let $\varnothing\ne\mathscr M\subset\mathscr X_N$ be compact and $U\in U(N)$. Then
\[
\dist_F(U,\mathcal C_{\mathscr M})=\min_{\mu\in\mathscr M}\delta_2\big(\spec U,\mu\big),
\qquad
\dist_{\mathrm{op}}(U,\mathcal C_{\mathscr M})=\min_{\mu\in\mathscr M}\delta_\infty\big(\spec U,\mu\big),
\]
and both minima are attained by elements of $\mathcal C_{\mathscr M}$ that commute with $U$. This applies to $\mathcal S_{\typeSU}$
($\mathscr M_{\typeSU}=\{\Lambda_{\typeSU}(a):a\in[0,\pi]\}$), to $\Gel(n)=\bigcup_{\typeSU}\mathcal S_{\typeSU}$, to $\Gtl(n)$, to $\GtlU(n)$, and to each phase-inclusive
sector $\mathcal S^U_{\typeU}$ (more generally, to any finite union of such sectors). For $n\ge2$ it does not apply to the full dictionary
$\GelU(n)$, which is dense and not closed (Proposition~\ref{prop:U2_spectral}(iv)); for $n=1$ this dictionary is $U(2)$.
\end{theorem}

\begin{proof}
For normal matrices $U,V$ the Hoffman--Wielandt inequality gives $\|U-V\|_F\ge\delta_2(\spec U,\spec V)$
\cite{HoffmanWielandt1953}, \cite[Ch.~VI]{Bhatia1997}; for unitary $U,V$ the Bhatia--Davis inequality gives
$\|U-V\|\ge\delta_\infty(\spec U,\spec V)$ \cite{BhatiaDavis1984}, \cite{Bhatia1997}. Conversely, write
$U=Q\diag(\lambda_1,\dots,\lambda_N)Q^{-1}$ and, for $\mu\in\mathscr M$ and a permutation $\pi$ optimal for the norm under consideration
(the optimal permutations for $\delta_2$ and $\delta_\infty$ may differ), put
$V:=Q\diag(\mu_{\pi(1)},\dots,\mu_{\pi(N)})Q^{-1}\in\mathcal C_{\mathscr M}$; then $\|U-V\|_F=\delta_2$, respectively $\|U-V\|=\delta_\infty$. The
minima over $\mathscr M$ exist because $\mathscr M$ is non-empty and compact and $\delta_2,\delta_\infty$ are continuous. Each $\mathcal S_{\typeSU}$
and each $\mathcal S^U_{\typeU}$ is spectrally defined by Theorems~\ref{thm:spectral} and~\ref{thm:spectral_U2} (the family
$\{\Lambda^U_{\typeU}(z_1,z_2)\}$ is the continuous image of the torus $(S^1)^2$), and $\Gtl$, $\GtlU$ by
Proposition~\ref{prop:spectral_consequences}(ii).
\end{proof}

For the two-level dictionaries the matching problem can be solved in closed form.

\begin{corollary}[Closed forms for the two-level dictionaries]\label{cor:closed_forms}
Let $N\ge4$ (for $N=2$ one has $\GtlU=U(2)$ and $\Gtl=SU(2)$, which are not covered by the formulas below). Let $U\in U(N)$ have eigenvalues
$\lambda_1,\dots,\lambda_N$, ordered so that $|1-\lambda_1|\ge|1-\lambda_2|\ge\cdots$, and write
$\theta_k=\Arg\lambda_k$. Then
\begin{align*}
\dist_F(U,\GtlU)^2&=\sum_{k\ge3}|1-\lambda_k|^2=4\sum_{k\ge3}\sin^2(\theta_k/2),
&
\dist_{\mathrm{op}}(U,\GtlU)&=|1-\lambda_3|,
\\
\dist_F(U,\Gtl)^2&=\min_{i\ne j}\Big[4-2|\lambda_i+\bar\lambda_j|+\sum_{k\ne i,j}|1-\lambda_k|^2\Big],
&
\dist_{\mathrm{op}}(U,\Gtl)&=\min_{i\ne j}\max\Big\{2\sin\tfrac{\omega_{ij}}{4},\ \max_{k\ne i,j}|1-\lambda_k|\Big\},
\end{align*}
where $\omega_{ij}:=|\Arg(\lambda_i\lambda_j)|\in[0,\pi]$ is the angular distance between $\lambda_i$ and $\bar\lambda_j$. For $\Gel$ the
distance is a minimum over the finitely many types $\typeSU$ of a one-variable minimization in $a\in[0,\pi]$.
\end{corollary}

\begin{proof}
For $\GtlU$, $\mathscr M=\{\{\mu_1,\mu_2,1^{(N-2)}\}:\mu_1,\mu_2\in S^1\}$: in an optimal matching two eigenvalues are matched to $\mu_1,\mu_2$ at zero cost
and the remaining $N-2$ to $1$; both costs are minimized by leaving out the two eigenvalues farthest from $1$. For $\Gtl$,
$\mathscr M=\{\{z,\bar z,1^{(N-2)}\}:|z|=1\}$; for a fixed pair $(i,j)$ matched to $(z,\bar z)$,
$|\lambda_i-z|^2+|\lambda_j-\bar z|^2=|\lambda_i-z|^2+|\bar\lambda_j-z|^2=4-2\operatorname{Re}\big(\bar z(\lambda_i+\bar\lambda_j)\big)$,
minimized over $|z|=1$ at $4-2|\lambda_i+\bar\lambda_j|$, while $\min_{|z|=1}\max\{|\lambda_i-z|,|\bar\lambda_j-z|\}$ is attained at the
midpoint of the shorter arc and equals $2\sin(\omega_{ij}/4)$.
\end{proof}

In particular, the Frobenius distance of $U$ to the phase-inclusive two-level dictionary is the $\ell^2$-norm of the chordal
distances to $1$ of all but the two eigenvalues farthest from $1$, and the phase-free distance differs from it only through the
pairing term $4-2|\lambda_i+\bar\lambda_j|$, which measures the failure of the removed pair to be complex conjugate.

\subsection{Intrinsic two-level lengths}

For $U\in U(N)$ let $\ell^U(U)$ (resp.\ $\ell^{SU}(U)$ for $U\in SU(N)$) be the least $k$ such that
$U=T_1\cdots T_k$ with $T_i\in\GtlU(n)$ (resp.\ $T_i\in\Gtl(n)$), with $\ell^U(I_N)=\ell^{SU}(I_N)=0$. These are factor counts,
not lengths.

\begin{theorem}[Intrinsic two-level length]\label{thm:intrinsic_length}
Let $U\in U(N)$ and $r:=\operatorname{rank}(U-I)$. Then
\[
\ell^U(U)=\lceil r/2\rceil\le N/2,
\]
and the minimum is attained by pairwise commuting factors supported on mutually orthogonal $2$-planes spanned by eigenvectors of $U$.
\end{theorem}

\begin{proof}
Lower bound: if $T_i\in\GtlU$ then $\operatorname{rank}(T_i-I)\le2$, and
$T_1\cdots T_k-I=\sum_{i=1}^kT_1\cdots T_{i-1}(T_i-I)$ has rank at most $2k$. Upper bound: write
$U=Q\diag(\lambda_1,\dots,\lambda_r,1,\dots,1)Q^{-1}$ with $\lambda_i\ne1$, and let $\mathbf q_i:=Q\mathbf e_i$. Group the indices $1,\dots,r$ into
$\lceil r/2\rceil$ pairs: $\lfloor r/2\rfloor$ pairs of indices in $\{1,\dots,r\}$ and, if $r$ is odd, the pair $\{r,r+1\}$, where
$r+1\le N$ because $N$ is even and hence $r<N$, and $\lambda_{r+1}=1$. For a pair $J$ let $U_J$ be the
operator equal to $U$ on the $2$-plane $\mathrm{span}\{\mathbf q_i:i\in J\}$ and to the identity on its orthogonal complement; then
$\operatorname{rank}(U_J-I)\le2$, so $U_J\in\GtlU$ by Proposition~\ref{prop:spectral_consequences}(ii). These operators commute and
their product is $U$.
\end{proof}

For $U\in SU(N)\setminus\{I_N\}$, the familiar bounds
$\lceil r/2\rceil\le\ell^{SU}(U)\le r-1$ follow from the exact formula in Theorem~\ref{thm:lsu}(ii).

\begin{remark}[Intrinsic versus coordinate factorization]\label{rem:intrinsic_vs_coordinate}
For the phase-inclusive count, the relevant rank-one factors are unitary pseudoreflections
$I_N+(z-1)\mathbf v\mathbf v^\dagger$, $z\in S^1$, $\|\mathbf v\|=1$.
Their minimal count is $\operatorname{rank}(U-I_N)$, by the rank bound and the spectral decomposition;
grouping them on orthogonal eigenvector planes gives the two-level count.
The length of $U(N)$ with respect to these pseudoreflections was determined by Dieudonn\'e
\cite{Dieudonne1955}. They differ from the involutive reflections of
\cite{DjokovicMalzan1979}, which govern the determinant-one count of Theorem~\ref{thm:lsu}(ii);
the orthogonal results of Cartan and Scherk \cite[Theorem~1]{Scherk1950},
\cite[pp.~20--22]{Dieudonne1948} and the indefinite case \cite{DjokovicMalzan1982} provide the
wider classical context.
The interest of $\ell^U$ here lies in the comparison with coordinate factorizations. The Givens
algorithm of Theorem~\ref{thm:givens} uses up to $N(N-1)/2$ two-level factors on coordinate planes, plus a diagonal factor; this is an
upper bound for that algorithm, not a proven minimum over coordinate factorizations. Theorem~\ref{thm:intrinsic_length} gives the exact
minimum, at most $N/2$, once the supports may range over $\Gr_2(\C^N)$: the difference lies in the \emph{choice of supports}, and the
complexity of $U$ is transferred to the continuous datum of its eigenframe. In the terminology of \cite{AcevedoFalco2026}, $\ell^U$ is an
intrinsic cost (minimized over all admissible supports), whereas the coordinate count is the cost of one particular implementation. When
$-1\notin\spec(U)$, i.e.\ when all principal eigenangles of $U$ lie in $(-\pi,\pi)$, the commuting factorization is realized by the single
path $t\mapsto\exp(t\sum_iX_i)$, $X_i$ the minimal logarithms of the factors, which is a minimizing geodesic of $U(N)$
(Section~\ref{sec:variational}).
\end{remark}

For determinant-one factors, the rank alone does not determine the intrinsic count.
We use the Cayley transform on the active support to construct maximal rank reductions.
Let $H$ be a finite-dimensional Hilbert space and $V\in U(H)$ with $1\notin\spec(V)$; put
\begin{equation}\label{eq:cayley}
C_V:=-\tfrac{i}{2}(V+I_H)(V-I_H)^{-1},
\qquad (V-I_H)^{-1}=-\tfrac12I_H+iC_V.
\end{equation}
Then $C_V$ is Hermitian, with eigenvalues $-\tfrac12\cot(\vartheta/2)$ for the eigenangles
$\vartheta\in(0,2\pi)$ of $V$.
The function $\vartheta\mapsto-\tfrac12\cot(\vartheta/2)$ is strictly increasing.
Thus, if the eigenangles are listed increasingly and the Cayley eigenvalues decreasingly,
$h_j=-\tfrac12\cot(\vartheta_{r+1-j}/2)$, where $r=\dim H$.
For a gate with $\operatorname{rank}(U-I_N)\ge2$, we apply this to
$H=\Ran(U-I_N)$ and $V=U|_H$; the restriction has no eigenvalue $1$.
The one-step criterion below is linear algebra in finite dimension and does not require
the ambient dimension to be a power of two; we retain $N=2^n$ for the qubit notation.

\begin{lemma}\label{lem:cayley_su2}
For $S\in SU(2)\setminus\{I_2\}$ the matrix $C_S$ is traceless, and every traceless Hermitian $2\times2$ matrix is $C_S$ for exactly one
$S\in SU(2)\setminus\{I_2\}$.
\end{lemma}

\begin{proof}
An element $S\ne I_2$ of $SU(2)$ has eigenvalues $e^{\pm i\varphi}$ with $\varphi\in(0,\pi]$, hence $1\notin\spec(S)$, and
$\cot(\varphi/2)+\cot(-\varphi/2)=0$. Conversely, if $H$ is Hermitian with eigenvalues $\pm c$, then $M:=-\tfrac12I+iH$ is normal and
invertible, and $S:=I+M^{-1}$ is normal with eigenvalues $(\tfrac12\pm ic)/(-\tfrac12\pm ic)$, which have modulus $1$ and product $1$; so
$S\in SU(2)\setminus\{I_2\}$ and $(S-I)^{-1}=M$, i.e.\ $C_S=H$. Uniqueness follows from $S=I+(-\tfrac12I+iC_S)^{-1}$.
\end{proof}

The count in part (ii) below is the two-level consequence of the classical reflection-length theorem
of Djokovi\'c--Malzan \cite[Theorem]{DjokovicMalzan1979}.
For completeness, we give a direct constructive proof using age and the Cayley parameterization in part (i).
Its eigenangle-merging step is the two-level counterpart of \cite[Lemma~4]{DjokovicMalzan1979}, where one reflection
replaces two angles $\alpha<\pi<\beta$ by $0$ and $\alpha+\beta-\pi$.

\begin{theorem}[Maximal rank reductions and two-level length]\label{thm:lsu}
Let $n\ge1$ and $N=2^n$. For $U\in U(N)$ put $r:=\operatorname{rank}(U-I)$.
If $r\ge2$, let $R:=\Ran(U-I)=\ker(U-I)^\perp$, $V:=U|_R$, and
\[
C_R:=-\tfrac{i}{2}(V+I_R)(V-I_R)^{-1}.
\]
Let $h_1\ge\cdots\ge h_r$ be the eigenvalues of $C_R$, and let
$0<\vartheta_1\le\cdots\le\vartheta_r<2\pi$ be the eigenangles of $V$, so that
$h_j=-\tfrac12\cot(\vartheta_{r+1-j}/2)$.
\begin{enumerate}
\item[(i)] For every $U\in U(N)$ with $r\ge2$, there is $T\in\Gtl(n)$ such that
$\operatorname{rank}(T^{-1}U-I)\le r-2$ if and only if
\[
h_{r-1}+h_r\le0\le h_1+h_2,
\qquad\text{equivalently}\qquad
\vartheta_1+\vartheta_2\le2\pi\le\vartheta_{r-1}+\vartheta_r.
\]
In that case $\operatorname{rank}(T^{-1}U-I)=r-2$. Every admissible $T$ has its
supporting plane in $R$, and the admissible factors are exactly
\[
T=I_N+F(S-I_2)F^\dagger,
\]
where $F:\C^2\to R\subset\C^N$ is an isometry,
$\Tr(F^\dagger C_RF)=0$, and $S\in SU(2)\setminus\{I_2\}$ is uniquely
determined by $C_S=F^\dagger C_RF$.
\item[(ii)] Let $U\in SU(N)$ and put
\[
\nu(U):=\frac1{2\pi}\sum_{j=1}^{r}\vartheta_j,\qquad \nu(I_N):=0.
\]
Then
\[
\ell^{SU}(U)=\max\{\nu(U),\,r-\nu(U)\}.
\]
In particular, the rank bound is attained exactly when
$\nu(U)\in\{\lfloor r/2\rfloor,\lceil r/2\rceil\}$.
For $N=4$ this gives $\ell^{SU}(U)=\lceil r/2\rceil$ unless
$r=4$ and $\sum_j\vartheta_j\ne4\pi$, when $\ell^{SU}(U)=3$.
\item[(iii)] The open set $\mathcal X:=\{U\in SU(N):1\notin\spec(U)\}$
has exactly $N-1$ connected components
\[
\mathcal X_k:=\{U\in\mathcal X:\nu(U)=k\},\qquad k=1,\ldots,N-1.
\]
Each is nonempty and path connected, and $\ell^{SU}$ is constant on it,
with value $\max\{k,N-k\}$. The central component has value $N/2$;
the extreme components have value $N-1$.
\end{enumerate}
\end{theorem}

\begin{proof}
(i) First we show that a rank reduction by two forces the supporting plane
of $T$ to lie in $R$. Put $A:=U-I_N$, $B:=T-I_N$ and $D:=U-T=A-B$.
If $\operatorname{rank}(D)\le r-2$, then
\[
r=\operatorname{rank}(A)
\le\operatorname{rank}(D)+\operatorname{rank}(B)
\le(r-2)+2=r.
\]
Thus $\operatorname{rank}(D)=r-2$ and $\operatorname{rank}(B)=2$.
Moreover $\Ran(A)\subseteq\Ran(D)+\Ran(B)$, and the latter space has
dimension at most $r$. It follows that
$\Ran(D)+\Ran(B)=\Ran(A)=R$; in particular $\Ran(B)\subseteq R$.
Since $T$ is normal and $\operatorname{rank}(T-I_N)=2$, its supporting
plane is unique and equals $\Ran(T-I_N)$
(Definition~\ref{def:two_level_subgroup}).
Both $U$ and $T$ act as the identity on $R^\perp$, so the problem reduces
to $V=U|_R$, for which $V-I_R$ is invertible.

Write $T=I_N+F(S-I_2)F^\dagger$ with $F:\C^2\to R$ an isometry.
The identity factor cannot give the required rank reduction, so $S\ne I_2$.
Put $A_R:=V-I_R$ and $B_2:=S-I_2$, both invertible.
For $\mathbf x\in R$,
\[
(A_R-FB_2F^\dagger)\mathbf x=0
\quad\Longleftrightarrow\quad
\mathbf x=A_R^{-1}F\mathbf y,\qquad
(B_2^{-1}-F^\dagger A_R^{-1}F)\mathbf y=0,
\]
where $\mathbf y=B_2F^\dagger\mathbf x$.
Since $A_R^{-1}F$ is injective,
\[
\dim\ker(V-T|_R)
=\dim\ker(B_2^{-1}-F^\dagger A_R^{-1}F)\le2.
\]
The rank is $r-2$ exactly when this $2\times2$ matrix is zero, or, by~\eqref{eq:cayley}, equivalently
$C_S=F^\dagger C_RF$. By Lemma~\ref{lem:cayley_su2}, such an $S$ exists
exactly when $\Tr(F^\dagger C_RF)=0$, and is then unique.
As $F$ varies over the connected Stiefel manifold of orthonormal
$2$-frames in $R$, this trace takes every value between
$h_{r-1}+h_r$ and $h_1+h_2$, by Ky Fan's principle
\cite[Ch.~III]{Bhatia1997}. This proves the Cayley criterion and the
description of all admissible factors. The converse construction acts
as the identity on $R^\perp$, so it also gives rank $r-2$ in $\C^N$.

For $0<\alpha,\beta<2\pi$,
\[
-\tfrac12\left(\cot\tfrac\alpha2+\cot\tfrac\beta2\right)
=
-\frac{\sin((\alpha+\beta)/2)}
       {2\sin(\alpha/2)\sin(\beta/2)}.
\]
The denominator is positive, and the expression is non-negative exactly
when $\alpha+\beta\ge2\pi$.
Since $-\tfrac12\cot(\vartheta/2)$ is strictly increasing on $(0,2\pi)$,
the Cayley criterion is equivalent to the stated eigenangle inequalities.

(ii) For any unitary $X$, let $\operatorname{age}(X)$ be the sum of its
eigenangles in $[0,2\pi)$ divided by $2\pi$.
The Chen--Ruan inequality, in the unitary formulation and elementary proof of
\cite[Theorem~2.3(i)]{GuralnickTiep2010Age}, implies
\[
\operatorname{age}(XY)\le\operatorname{age}(X)+\operatorname{age}(Y).
\]
Indeed, the common fixed subspace of $X,Y$ is contained in that of $XY$,
so the fixed-space term in that theorem has the required sign.
For $X\in SU(N)$, $\operatorname{age}(X)=\nu(X)$ is an integer, and
\[
\nu(X^{-1})=\operatorname{rank}(X-I_N)-\nu(X).
\]
Every nonidentity factor in $\Gtl$ has age $1$, as does its inverse.
Applying subadditivity to a factorization and to its inverse therefore gives
\[
\ell^{SU}(U)\ge\max\{\nu(U),\,r-\nu(U)\}.
\]

We prove the reverse bound constructively by induction on $r$.
The identity is immediate. A nonidentity element of $SU(N)$ has
$r\ge2$ and $1\le\nu(U)\le r-1$.
For $r=2$ the gate itself lies in $\Gtl$ and $\nu(U)=1$.
Write $\nu:=\nu(U)$.

First suppose $2\nu=r$.
The sum of the active eigenangles is $2\pi\nu=r\pi$, so their ordering gives
$\vartheta_1+\vartheta_2\le2\pi\le\vartheta_{r-1}+\vartheta_r$.
By (i), there is $T\in\Gtl$ such that $W:=T^{-1}U$ has active rank $r-2$.
Subadditivity for $U=TW$ and $U^{-1}=W^{-1}T^{-1}$ yields
\[
\nu-1\le\nu(W),
\qquad
r-\nu\le (r-2-\nu(W))+1,
\]
hence $\nu(W)=\nu-1$.
The induction hypothesis realizes $W$ with $\nu-1$ factors, so $U$ uses $\nu$.

Next suppose $2\nu<r$.
The two smallest active eigenangles satisfy
$0<\vartheta_1+\vartheta_2<2\pi$.
Choose orthonormal eigenvectors $\mathbf q_1,\mathbf q_2$ for them and let
$T\in\Gtl$ act by $\diag(e^{i\vartheta_1},e^{-i\vartheta_1})$
on their plane and by the identity on its orthogonal complement.
Then $W:=T^{-1}U$ replaces this eigenvalue pair by
$1,e^{i(\vartheta_1+\vartheta_2)}$, so
$\operatorname{rank}(W-I_N)=r-1$ and $\nu(W)=\nu$.
By induction $W$ uses at most
$\max\{\nu,r-1-\nu\}=r-1-\nu$ factors, and $U$ uses at most $r-\nu$.

Finally, if $2\nu>r$, apply the preceding construction to $U^{-1}$,
whose age is $r-\nu<r/2$, and invert its factors.
Since $\ell^{SU}(U^{-1})=\ell^{SU}(U)$, this gives $\nu$ factors.
Together with the lower bound, this proves the formula.
The statements about the rank bound and $SU(4)$ follow immediately.

(iii) On $\mathcal X$, the positive eigenangle logarithm $H$ defined by
$U=e^{iH}$, $0<H<2\pi I_N$, depends continuously on $U$.
Thus $\nu(U)=\Tr(H)/(2\pi)$ is locally constant and lies in
$\{1,\ldots,N-1\}$. Each $\mathcal X_k$ contains $e^{2\pi i k/N}I_N$.
For $U=e^{iH}\in\mathcal X_k$, the path
\[
[0,1]\ni t\longmapsto
\exp\!\left(i\left[(1-t)H+t\frac{2\pi k}{N}I_N\right]\right)
\]
stays in $\mathcal X_k$ and joins $U$ to that scalar matrix.
Consequently the sets $\mathcal X_k$ are exactly the connected components.
The length statement follows from (ii), since $r=N$.
\end{proof}

\begin{remark}\label{rem:lsu}
The exact length separates the rank from the angular distribution:
the rank bound is attained precisely when
$\nu(U)\in\{\lfloor r/2\rfloor,\lceil r/2\rceil\}$.
For even $r$, the one-step criterion is necessary but not sufficient for this equality.
For example, $U=\diag(e^{i\pi/6}I_6,-iI_2)\in SU(8)$ has $r=8$ and $\nu(U)=2$.
Its two smallest active angles sum to $\pi/3$ and its two largest to $3\pi$, so the criterion holds,
but $\ell^{SU}(U)=6$.
Every maximal reduction leaves rank $6$ and age $1$; the positive active angles then sum to $2\pi$,
so the criterion fails at the next step.

Odd rank must be distinguished.
With $\omega=e^{2\pi i/3}$, the gate $U=\diag(\omega I_3,1)$ has no rank reduction by two, yet
\[
U=\diag(\omega,\omega^{-1},1,1)\,
  \diag(1,\omega^2,\omega,1)
\]
has length $2=\lceil3/2\rceil$.
Finally, $\ell^{SU}(e^{2\pi i/N}I_N)=N-1$.
Thus the bound $\lceil r/2\rceil+1$ is false in general, and the extra number of determinant-one factors need not be one.
\end{remark}

\subsection{Finite alphabets and covariance}

For a finite set $\mathcal A\subset U(N)$ put $\mathcal A^{-1}:=\{a^{-1}:a\in\mathcal A\}$ and let $\langle\mathcal A\rangle$ be the subgroup it
generates. A \emph{word} over $\mathcal A$ is a finite sequence $w=(a_1,\dots,a_L)$ of letters in $\mathcal A\cup\mathcal A^{-1}$, of length
$|w|:=L$, with evaluation $\evw(w):=a_1\cdots a_L\in U(N)$ (the empty word has length $0$ and evaluation $I_N$). In norm estimates we
write $w$ for $\evw(w)$; lengths always refer to words, since a matrix may be the evaluation of words of different lengths.
For a finite set $G_2\subset SU(2)$ we call
\[
\mathcal A(n;G_2):=\bigcup_{1\le p<q\le N}\phi_{p,q}(G_2)
\]
the \emph{coordinate two-level alphabet}; if $\langle G_2\rangle$ is dense in $SU(2)$, then $\langle\mathcal A(n;G_2)\rangle$ is dense in
$SU(N)$, and Appendix~\ref{app:compilation} gives explicit compilation bounds.

\begin{remark}[Finite alphabets are necessarily extrinsic]\label{rem:finite_extrinsic}
If a finite set $\mathcal A\subset U(N)$ is invariant under conjugation by $U(N)$, then $\mathcal A$ consists of scalars: the conjugacy
class of $A\in\mathcal A$ is connected and finite, hence a point. Such an alphabet generates a group of scalars and therefore cannot be
universal for $N\ge2$. Consequently every universal finite-alphabet compilation breaks conjugation
invariance under the full group $U(N)$; what can survive is \emph{equivariance}: if $\|U-\evw(w)\|\le\varepsilon$, then conjugating
each letter gives a word $QwQ^{-1}$ of the same length over $Q\mathcal AQ^{-1}$ with $\|QUQ^{-1}-\evw(QwQ^{-1})\|\le\varepsilon$, and,
for alphabets built from geometric data such as a finite family of planes with chosen frames, \emph{covariance} under the finite subgroup
preserving those data. The next proposition shows that the coordinate alphabet is an instance of the latter.
\end{remark}

For $r\in\N$ let $\Gamma_r(N)\subset U(N)$ be the group of monomial matrices whose non-zero entries lie in $\mu_r$, so that
$\Gamma_r(N)\cong\mu_r^N\rtimes\mathfrak S_N$ and $\Gamma_1(N)$ is the group of permutation matrices; put $H_r:=\Gamma_r(2)$, a group of
order $2r^2$ containing $X$.

\begin{proposition}[Covariance of the coordinate alphabet]\label{prop:covariance}
Let $G_2\subset SU(2)$ be finite and $r\in\N$.
\begin{enumerate}
\item[(i)] Let $M\in\Gamma_r(N)$, with underlying permutation $\sigma$, and $1\le p<q\le N$. Then
$M\phi_{p,q}(S)M^{-1}=\phi_{p',q'}(hSh^{-1})$ for all $S\in SU(2)$, where $\{p',q'\}=\{\sigma(p),\sigma(q)\}$, $p'<q'$, and $h\in H_r$
depends only on $M,p,q$. Consequently, if $hG_2h^{-1}=G_2$ for all $h\in H_r$, then
$M\mathcal A(n;G_2)M^{-1}=\mathcal A(n;G_2)$ for all $M\in\Gamma_r(N)$.
\item[(ii)] For every finite $G_2$ the set $G_2^{(r)}:=\{hgh^{-1}:h\in H_r,\ g\in G_2\}\subset SU(2)$ is finite, contains $G_2$ and is invariant
under conjugation by $H_r$; in particular $\langle G_2^{(r)}\rangle$ is dense in $SU(2)$ whenever $\langle G_2\rangle$ is.
\item[(iii)] For $U\in U(N)$ and $\varepsilon>0$ let
\[
L_\varepsilon(U):=\inf\{|w|:\ w\text{ a word over }\mathcal A(n;G_2),\ \|U-\evw(w)\|\le\varepsilon\},
\]
with $\inf\varnothing:=+\infty$ (when the set is non-empty the infimum is a minimum, lengths being integers), and let
$L^{\mathrm{ph}}_\varepsilon(U)$ be defined in the same way with the condition $\min_{\theta\in\R}\|U-e^{i\theta}\evw(w)\|\le\varepsilon$.
Under the invariance hypothesis of (i), $L_\varepsilon(MUM^{-1})=L_\varepsilon(U)$ and $L^{\mathrm{ph}}_\varepsilon(MUM^{-1})=
L^{\mathrm{ph}}_\varepsilon(U)$ for all $U\in U(N)$, $\varepsilon>0$ and $M\in\Gamma_r(N)$, the values $+\infty$ included. (Both costs can
be infinite: words evaluate in $SU(N)$, and if $\langle G_2\rangle$ is not dense, admissible words may not exist at all.)
\end{enumerate}
Thus $\mathcal A(n;G_2^{(r)})$ is the alphabet $\bigcup_{W\in\mathcal W}\phi_{W,f_W}(G_2^{(r)})$ attached to the $\Gamma_r(N)$-orbit
$\mathcal W=\{W_{p,q}\}$ of the coordinate plane $W_{1,2}$ with the coordinate frames, and it is $\Gamma_r(N)$-covariant.
\end{proposition}

\begin{proof}
(i) Write $M=DP_\sigma$ with $P_\sigma \mathbf e_p=\mathbf e_{\sigma(p)}$ and $D=\diag(d_1,\dots,d_N)$, $d_i\in\mu_r$, so that $M\mathbf e_p=d_{\sigma(p)}\mathbf e_{\sigma(p)}$.
Then $M$ maps $W_{p,q}$ onto $W_{\sigma(p),\sigma(q)}$ and $W_{p,q}^\perp$ onto $W_{\sigma(p),\sigma(q)}^\perp$, so $M\phi_{p,q}(S)M^{-1}$ is
the identity on $W_{\sigma(p),\sigma(q)}^\perp$ and acts on $W_{\sigma(p),\sigma(q)}$, in the frame $(\mathbf e_{\sigma(p)},\mathbf e_{\sigma(q)})$, by
$h_0Sh_0^{-1}$ with $h_0=\diag(d_{\sigma(p)},d_{\sigma(q)})\in H_r$. If $\sigma(p)<\sigma(q)$ take $h=h_0$; otherwise the coordinate frame
of $W_{\sigma(q),\sigma(p)}$ is the reversed frame and $h=Xh_0\in H_r$. The inclusion $M\mathcal AM^{-1}\subseteq\mathcal A$ follows, and
applying it to $M^{-1}$ gives equality. (ii) $H_r$ is finite, and conjugation by $U(2)$ preserves $SU(2)$. (iii) $w\mapsto MwM^{-1}$
(letterwise) is a length-preserving bijection of words over $\mathcal A$, and $\|MUM^{-1}-\evw(MwM^{-1})\|=\|U-\evw(w)\|$; global phases
commute with $M$. Hence it maps the admissible words for $U$ bijectively onto those for $MUM^{-1}$, and the two infima coincide, also when
both sets are empty.
\end{proof}

For instance, with $r=1$ it suffices that $G_2$ be closed under $S\mapsto XSX^{-1}$, and then the coordinate alphabet is
$\mathfrak S_N$-covariant: relabelling the computational basis does not change compiled word lengths. Which other finite subgroups
$\Gamma\subset U(N)$ admit universal $\Gamma$-covariant alphabets of this geometric form is discussed in Section~\ref{sec:conclusion}.

\section{Metric structure of the embedded subgroups}
\label{sec:variational}

Equip $U(N)$ with the bi-invariant metric induced by $\langle\cdot,\cdot\rangle_{\mathrm{hs}}$. For an $H^1$ curve
$U:[0,1]\to U(N)$ (a matrix-valued $H^1$ function with values in $U(N)$, identified with its continuous representative) put
$A(t):=\dot U(t)U(t)^{-1}$, defined for almost every $t$, and
$\mathcal E[U]:=\tfrac12\int_0^1\|A\|_{\mathrm{hs}}^2\,dt$ (energy), $\mathcal L[U]:=\int_0^1\|A\|_{\mathrm{hs}}\,dt$ (length).

\begin{lemma}[{\cite[\S21]{Milnor1963}}]\label{lem:geodesics}
Let $G$ be a Lie group with a bi-invariant metric, and for $X\in\mathfrak g$ let $X^L$, $X^R$ be the left- and right-invariant
vector fields with value $X$ at $e$. Then $\nabla_{X^L}Y^L=\tfrac12[X,Y]^L$ for left-invariant fields and
$\nabla_{X^R}Y^R=-\tfrac12[X,Y]^R$ for right-invariant fields; the geodesics are the curves $g_0\exp(tX)$; and every Lie subgroup
$H\le G$ is totally geodesic.
\end{lemma}

\begin{proof}
The formulas for $\nabla$ follow from Koszul's formula and $\Ad$-invariance; note $[X^R,Y^R]=-[X,Y]^R$. Hence $\nabla_{X^L}X^L=0$,
so $t\mapsto g_0\exp(tX)$ is a geodesic, and these exhaust the geodesics by uniqueness of initial-value problems. If $g_0\in H$ and
$X\in\Lie(H)$, the geodesic $g_0\exp(tX)$ stays in $H$.
\end{proof}

Fix $\typeSU\in\mF(N)$, $\phi\in\mathcal O_{\typeSU}$ and let $K:=K_\phi=\phi(SU(2))$, with the induced metric and with Lie algebra
$\mathfrak k=d\phi_e(\mathfrak{su}(2))$. For $U_\star\in K$ put $\Log_{\mathfrak k}(U_\star):=\{X\in\mathfrak k:\exp X=U_\star\}\ne\varnothing$.

\begin{theorem}[Constrained minimizers]\label{thm:constrained}
Let $U_\star\in K$ and $\mathcal C_{I,U_\star}(K):=\{U\in H^1([0,1],K):U(0)=I_N,\ U(1)=U_\star\}$. Then the following statements hold.
\begin{enumerate}
\item[(i)] $\mathcal E$ attains its minimum on $\mathcal C_{I,U_\star}(K)$, and $\min\mathcal E=\tfrac12 d_K(I_N,U_\star)^2
=\tfrac12\min\{\|X\|_{\mathrm{hs}}^2:X\in\Log_{\mathfrak k}(U_\star)\}$.
\item[(ii)] $U\in\mathcal C_{I,U_\star}(K)$ is a minimizer if and only if $U(t)=\exp(tX)$ with $X$ a minimal-norm element of
$\Log_{\mathfrak k}(U_\star)$.
\item[(iii)] Every minimizer is a geodesic of $U(N)$. An ambient geodesic $\exp(tX)$ with $X\in\mathfrak k$ and $\exp X=U_\star$ is a
constrained minimizer \emph{if and only if} $\|X\|_{\mathrm{hs}}$ is minimal in $\Log_{\mathfrak k}(U_\star)$; in general it is
not: for $U_\star=\phi(g_a)$ with $0<a<\pi$, the element $X=d\phi_e\big(\diag(i(a-2\pi),-i(a-2\pi))\big)$ satisfies $\exp X=U_\star$ and
$\|X\|_{\mathrm{hs}}=\sqrt{c_{\typeSU}}\,(2\pi-a)>\sqrt{c_{\typeSU}}\,a$ (Proposition~\ref{prop:dynkin}).
\end{enumerate}
\end{theorem}

\begin{proof}
(i)--(ii) $K$ is a compact connected Riemannian manifold; by Hopf--Rinow there is a minimizing geodesic $\gamma$ from $I$ to $U_\star$,
parametrized with constant speed $d_K(I_N,U_\star)$. For $U\in\mathcal C_{I,U_\star}(K)$, Cauchy--Schwarz gives
$2\mathcal E[U]\ge\mathcal L[U]^2\ge d_K(I_N,U_\star)^2=2\mathcal E[\gamma]$, so $\gamma$ minimizes $\mathcal E$; equality forces
constant speed and $\mathcal L[U]=d_K(I_N,U_\star)$, so every minimizer is a constant-speed minimizing geodesic
\cite[Ch.~9]{doCarmoRiemannian}. By Lemma~\ref{lem:geodesics} applied to $K$ (the induced metric is bi-invariant), these are the curves
$\exp(tX)$, $X\in\Log_{\mathfrak k}(U_\star)$, of energy $\tfrac12\|X\|_{\mathrm{hs}}^2$; minimality selects the minimal-norm $X$. (iii)
The first claim is Lemma~\ref{lem:geodesics}; the second follows from (ii).
\end{proof}

\begin{proposition}[Metric scaling: a Dynkin index]\label{prop:dynkin}
Let $\typeSU\in\mF(N)$, $\phi\in\mathcal O_{\typeSU}$, and let $c_{\typeSU}$ be as in \eqref{eq:cm_dmax}. Then
$\|d\phi_e(A)\|_{\mathrm{hs}}^2=c_{\typeSU}\|A\|_{\mathrm{hs}}^2$ for all $A\in\mathfrak{su}(2)$. Consequently the subgroup $K_\phi$, with the
induced metric, is isometric to the round $3$-sphere of radius $\sqrt{c_{\typeSU}}$; $d_{K_\phi}(I_N,\phi(g_a))=\sqrt{c_{\typeSU}}\,a$ for $a\in[0,\pi]$; the cut locus of
$I_N$ in $K_\phi$ is $\{\phi(-I_2)\}$ (in general $\phi(-I_2)\ne-I_N$: for the two-level type it acts as $-I$ on the supporting
plane and as $I$ on its complement); and the minimizer in Theorem~\ref{thm:constrained} is unique if and only if $U_\star\ne\phi(-I_2)$.
\end{proposition}

\begin{proof}
Both sides are $\Ad$-invariant, so it suffices to take $A=\diag(ia,-ia)$, $\|A\|_{\mathrm{hs}}^2=a^2$. By \eqref{eq:weights},
$\|d\rho_d(A)\|_{\mathrm{hs}}^2=\tfrac12a^2\sum_{j=0}^d(d-2j)^2=a^2\binom{d+2}{3}$, and the contributions of the isotypic blocks add. For
$c_{\typeSU}=1$ the metric on $SU(2)$ is the round metric of the unit sphere $S^3\subset\mathbb H$ (one-parameter subgroups with
$\|A\|_{\mathrm{hs}}=1$ are great circles of length $2\pi$); scaling by $c_{\typeSU}$ gives radius $\sqrt{c_{\typeSU}}$. Since $\phi$ is faithful,
$\phi(g_a)$, $a\in[0,\pi]$, lies at distance $\sqrt{c_{\typeSU}}\,a$ from $I$, and the cut locus of a point of a round sphere is its antipode
$\phi(g_\pi)=\phi(-I_2)$, which is joined to $I_N$ by infinitely many minimizing geodesics; every other point by exactly one.
\end{proof}

\begin{proposition}[When embedded subgroups are minimizing in $U(N)$]\label{prop:ambient_minimizing}
Let $\typeSU\in\mF(N)$, $\phi\in\mathcal O_{\typeSU}$ and $a\in[0,\pi]$, and let $d_{\max}=d_{\max}(\typeSU)$. Then
\[
d_{U(N)}\big(I_N,\phi(g_a)\big)\le d_{K_\phi}\big(I_N,\phi(g_a)\big)=\sqrt{c_{\typeSU}}\,a,
\]
with equality if and only if $d_{\max}\,a\le\pi$. In particular, for every type with $d_{\max}(\typeSU)=1$, i.e.\ $m_d=0$ for every $d\ge2$ (this includes
the two-level type and the tensor placements of type $V_1^{\oplus N/2}$), the embedded subgroups are metrically optimal in $U(N)$ for
every angle, whereas for the spin-$\tfrac32$ type in $U(4)$ their geodesics cease to minimize in $U(4)$ at $a=\pi/3$.
\end{proposition}

\begin{proof}
For $V\in U(N)$ with principal eigenangles $\theta_k=\Arg\lambda_k$,
$d_{U(N)}(I_N,V)^2=\min\{\|X\|_{\mathrm{hs}}^2:X\in\mathfrak u(N),\ \exp X=V\}=\tfrac12\sum_k\theta_k^2$, attained at $X_V$: every
skew-Hermitian logarithm of $V$ commutes with $V$ and has eigenvalues $i(\theta_k+2\pi n_k)$, $n_k\in\Z$. The eigenangles of $\phi(g_a)$ are
$(d-2j)a$ modulo $2\pi$, and $\tfrac12\sum_{d,j}m_d(d-2j)^2a^2=c_{\typeSU}a^2$ by Proposition~\ref{prop:dynkin}. The principal representative
of $(d-2j)a$ has strictly smaller modulus if and only if $|d-2j|a>\pi$, which occurs for some $j$ if and only if $d_{\max}a>\pi$.
\end{proof}

\begin{remark}[Comparison with geometric complexity]\label{rem:nielsen}
In \cite{NielsenDowlingGuDoherty2006,Nielsen2006,DowlingNielsen2008} circuit complexity is the geodesic distance from $I_N$ for a
right-invariant metric on $SU(N)$ that weights directions of the Lie algebra by penalties; see \cite[\S3]{AcevedoFalco2026} for the
penalty-matrix formulation. We use the following self-contained version. Fix a $\langle\cdot,\cdot\rangle_{\mathrm{hs}}$-orthonormal basis
$(B_\alpha)$ of $\mathfrak{su}(N)$ and a symmetric positive definite matrix $\Omega$, put $\|X\|_\Omega^2:=x^\top\Omega x$ for
$X=\sum_\alpha x_\alpha B_\alpha$, extend it right-invariantly, and let $G_\Omega(U)$ be the induced distance from $I_N$ to $U\in SU(N)$.
The isotropic case $\Omega=I$ is the bi-invariant metric \cite[\S3.2]{AcevedoFalco2026}; the locality penalty of
\cite[Eq.~(20)]{AcevedoFalco2026} is $\Omega_q=I\oplus qI$, $q>1$, with respect to the splitting of $\mathfrak{su}(N)$ into the span
of the skew-Hermitian generators $B_P:=-i\sqrt{2/N}\,P$, $\|B_P\|_{\mathrm{hs}}=1$, of the Pauli strings $P\ne I_N$ with one or two
non-identity factors, and its $\langle\cdot,\cdot\rangle_{\mathrm{hs}}$-orthogonal complement.
If $\Omega\ge I$ (weights at least $1$; no normalization by $\Tr\Omega$), then $\|X\|_\Omega\ge\|X\|_{\mathrm{hs}}$, every curve is at
least as long for $\Omega$, and for $U\in SU(N)$ with principal eigenangles $\theta_k$,
\[
G_\Omega(U)\ \ge\ d_{SU(N)}(I_N,U)\ \ge\ d_{U(N)}(I_N,U)=\Big(\tfrac12\sum_k\theta_k^2\Big)^{1/2}
\ \ \Big(=\sqrt{c_{\typeSU}}\,a\ \text{ if } (\typeSU,a)\in\Sigma(U)\text{ and } d_{\max}(\typeSU)\,a\le\pi\Big)
\]
by Proposition~\ref{prop:ambient_minimizing}. (The same argument applies on $U(N)$ once a weight $\ge1$ is assigned to the central
direction $i\R I_N$.) Thus the spectral data of Section~\ref{sec:spectral} give explicit lower bounds for locality-penalized complexity,
but the isotropic layer itself cannot distinguish local from entangling gates (Proposition~\ref{prop:elem_vs_entangling}). Within this
layer, for $U\in\Gel$ define
\[
\Lsu(U):=\min_{(\typeSU,a)\in\Sigma(U)}\sqrt{c_{\typeSU}}\,a,
\]
the minimal length of a path from $I_N$ to $U$ inside a single embedded subgroup $K_\phi$ containing $U$ (Proposition~\ref{prop:dynkin});
the corresponding minimal energy in unit time is $\tfrac12\Lsu(U)^2$. It is a length, not a factor count such as $\ell^U$ or
$\ell^{SU}$, and it is defined only on $\Gel$ (for CNOT, $\Sigma(U)=\varnothing$). One has $\Lsu(U)\ge d_{U(N)}(I_N,U)$, with equality
for $U\in\Gtl$.
\end{remark}

For the next result, let $r(U):=\min\{|w|:w\in\conv\spec(U)\}$ be the distance from $0$ to the convex hull of the eigenvalues of $U$, and
\[
\kappa(U):=\big(2(1-r(U)^2)\big)^{1/2}.
\]
Both are spectral invariants. $r(U)>0$ if and only if $\spec(U)$ lies in an open arc of length $<\pi$; in that case, if $w(U)$ denotes the
length of the shortest closed arc containing $\spec(U)$, then $r(U)=\cos(w(U)/2)$ and $\kappa(U)=\sqrt2\sin(w(U)/2)$ (after a global phase
the hull lies in the circular segment cut off by the chord joining the two end eigenvalues, and the midpoint of that chord is the point
of the hull closest to $0$).

\begin{proposition}[Spectral bound on coherence production]\label{prop:coherence}
Fix an orthonormal basis $\mathcal B=(\mathbf b_1,\dots,\mathbf b_N)$, let $P_j:=\mathbf b_j\mathbf b_j^\dagger$, let $\Delta_{\mathcal B}(A):=\sum_jP_jAP_j$ be the
complete dephasing channel, and let $C_{\mathcal B}(\rho):=\|\rho-\Delta_{\mathcal B}(\rho)\|_F^2$ be the (linear-entropy) coherence of a
state $\rho$ ($\rho=\rho^\dagger\ge0$, $\Tr\rho=1$). For $U\in U(N)$ put
\[
\mathfrak C(U):=\sup_{\mathcal B,\,\rho}\Big|\sqrt{C_{\mathcal B}(U\rho U^\dagger)}-\sqrt{C_{\mathcal B}(\rho)}\Big| .
\]
Then
\[
r(U)\,\kappa(U)\ \le\ \mathfrak C(U)\ \le\ \kappa(U)\ \le\ 2\|U-I_N\|\ \le\ 2\sqrt2\,d_{U(N)}(I_N,U),
\]
and the lower bound is attained by an incoherent pure state. In particular $\mathfrak C(U)/\kappa(U)\to1$ as $U\to I_N$ with
$U\notin S^1I_N$, so $\kappa$ is sharp to first order. If $(\typeSU,a)\in\Sigma(U)$ and $d_{\max}(\typeSU)\,a\le\pi/2$, then
\[
\tfrac{\sqrt2}{2}\sin\big(2d_{\max}(\typeSU)\,a\big)\ \le\ \mathfrak C(U)\ \le\ \sqrt2\,\sin\big(d_{\max}(\typeSU)\,a\big);
\]
for $U\in\Gtl$ with $(\typeTwo,a)\in\Sigma(U)$, $a\le\pi/2$, the upper bound is $\sqrt2\sin a$.
\end{proposition}

\begin{proof}
\emph{Upper bound.} Since $\Delta_{\mathcal B}$ is an orthogonal projection on $M_N(\C)$ for the Frobenius inner product, so is
$\mathrm{id}-\Delta_{\mathcal B}$, and $\sqrt{C_{\mathcal B}(\rho)}=\|(\mathrm{id}-\Delta_{\mathcal B})\rho\|_F$. By the reverse triangle
inequality the quantity in $\mathfrak C(U)$ is at most $\|U\rho U^\dagger-\rho\|_F$, a convex function of $\rho$; its supremum over the
convex set of states is attained at a pure state $\rho=\bm{\psi}\bm{\psi}^\dagger$, for which
$\|U\rho U^\dagger-\rho\|_F^2=2-2|\langle\bm{\psi},U\bm{\psi}\rangle|^2$. Writing $U=\sum_k\lambda_k\mathbf q_k\mathbf q_k^\dagger$,
$\langle\bm{\psi},U\bm{\psi}\rangle=\sum_k\lambda_k|\langle \mathbf q_k,\bm{\psi}\rangle|^2\in\conv\spec(U)$, and every point of the hull arises in this way; hence
$\sup_\rho\|U\rho U^\dagger-\rho\|_F=\kappa(U)$. Moreover $\|U\rho U^\dagger-\rho\|_F\le\|(U-I_N)\rho U^\dagger\|_F+\|\rho(U^\dagger-I_N)\|_F\le
2\|U-I_N\|\,\|\rho\|_F\le2\|U-I_N\|$, so $\kappa(U)\le2\|U-I_N\|$. Finally
$\|U-I_N\|=\max_k|1-e^{i\theta_k}|=2\max_k|\sin(\theta_k/2)|\le\max_k|\theta_k|\le\|X_U\|_F=\sqrt2\,d_{U(N)}(I_N,U)$ (proof of
Proposition~\ref{prop:ambient_minimizing}).

\emph{Lower bound.} If $r(U)=0$ there is nothing to prove. Otherwise let $\lambda_A,\lambda_B$ be the end eigenvalues of the shortest arc
containing $\spec(U)$, of length $w=w(U)<\pi$, with eigenvectors $\mathbf q_A,\mathbf q_B$. Let $\mathcal B$ consist of $\mathbf b_1=(\mathbf q_A+\mathbf q_B)/\sqrt2$,
$\mathbf b_2=(\mathbf q_A-\mathbf q_B)/\sqrt2$ and the remaining eigenvectors of $U$, and $\rho=\mathbf b_1\mathbf b_1^\dagger$, so $C_{\mathcal B}(\rho)=0$. Then
$U\mathbf b_1=(\lambda_A\mathbf q_A+\lambda_B\mathbf q_B)/\sqrt2$ has coordinates $(\lambda_A+\lambda_B)/2$ and $(\lambda_A-\lambda_B)/2$ on $\mathbf b_1,\mathbf b_2$ and $0$
elsewhere, with squared moduli $p_1=\cos^2(w/2)$, $p_2=\sin^2(w/2)$. For a pure state $\sigma$,
$C_{\mathcal B}(\sigma)=\|\sigma\|_F^2-\|\Delta_{\mathcal B}\sigma\|_F^2=1-\sum_jp_j^2$, so
$C_{\mathcal B}(U\rho U^\dagger)=2p_1p_2=\tfrac12\sin^2w$ and $\sqrt{C_{\mathcal B}(U\rho U^\dagger)}=\sqrt2\sin(w/2)\cos(w/2)=\kappa(U)\,r(U)$.

\emph{Limits and descriptors.} As $U\to I_N$, $r(U)\to1$. If $(\typeSU,a)\in\Sigma(U)$ and $d_{\max}a\le\pi/2$, the spectrum lies in the arc
$\{e^{i\theta}:|\theta|\le d_{\max}a\}$ and contains its endpoints, so $w(U)=2d_{\max}a$ when $d_{\max}a<\pi/2$, giving
$r=\cos(d_{\max}a)$ and $\kappa=\sqrt2\sin(d_{\max}a)$; for $d_{\max}a=\pi/2$ both sides reduce to $0\le\mathfrak C(U)\le\sqrt2$.
\end{proof}

\begin{remark}
The weaker bound $2(\sum_k\theta_k^2)^{1/2}$ is what one obtains by integrating the coherence growth inequality of
\cite[Prop.~4]{AcevedoFalco2026} along the minimizing geodesic $s\mapsto e^{sX_U}$. For a single gate, Proposition~\ref{prop:coherence}
replaces it by the endpoint quantity $\kappa(U)$, which depends only on the angular spread of $\spec(U)$ and is attained up to the factor
$r(U)$. The coherence change for a fixed reference basis depends on that basis,
whereas the supremum $\mathfrak C(U)$ over all bases is conjugation invariant.
The spectral bound $\kappa(U)$ controls the change for every fixed basis;
for $U\in\Gel$ it is controlled by the type through $d_{\max}(\typeSU)$ and the angle $a$.
This is one concrete sense in which the intrinsic layer constrains extrinsic, operational resources.
\end{remark}

\begin{remark}[$U(2)$-embeddings]
Let $\Phi\in\mathcal O^U_{\typeU}$, $\typeU\in\mFU(N)$, and let $\typeSU$ be the type of $\Phi|_{SU(2)}$ (Proposition~\ref{prop:U2_types}(iii)). The Lie
algebra $\mathfrak u(2)=\mathfrak{su}(2)\oplus i\R I_2$ splits into commuting ideals, orthogonal for every $\Ad$-invariant inner product,
and by \eqref{eq:weights_U2}, $d\Phi_e(i\tau I_2)$ acts on each $V_{d,k}$ as $i(d+2k)\tau$. Since $d\rho_{d,k}(A)$ is traceless for
$A\in\mathfrak{su}(2)$, the cross terms vanish and
\[
\|d\Phi_e(A+i\tau I_2)\|_{\mathrm{hs}}^2=c_{\typeSU}\|A\|_{\mathrm{hs}}^2+b_{\typeU}\tau^2,
\qquad
b_{\typeU}:=\tfrac12\sum_{d,k}\mu_{d,k}(d+1)(d+2k)^2>0,
\]
so that phase and non-abelian costs separate; $b_{\typeU}$ depends on the determinant twists, not only on $\typeSU$.
\end{remark}

\section{Conclusions and open problems}
\label{sec:conclusion}

For two-level $U(2)$ factors the intrinsic count is determined by the rank:
$\ell^U(U)=\lceil\operatorname{rank}(U-I)/2\rceil$.
For determinant-one factors, Theorem~\ref{thm:lsu} gives
\[
\ell^{SU}(U)=\max\{\nu(U),\operatorname{rank}(U-I)-\nu(U)\},
\]
where $\nu(U)$ is the positive eigenangle sum divided by $2\pi$.
This count is a two-level consequence of the classical reflection-length theorem of
Djokovi\'c--Malzan \cite[Theorem]{DjokovicMalzan1979}.
The proof given here realizes it through eigenangle merging and traceless Cayley compressions on the active support.
The Cayley parameterization additionally describes all factors giving a maximal one-step rank reduction.
The one-step criterion holds at every active rank, but does not generally propagate outside the balanced regime.
On the full-rank open set, the $N-1$ connected components are indexed by age $k$,
and the length on the $k$th component is $\max\{k,N-k\}$.
The classification of $SU(4)$ is a special case of the formula; in full rank, two factors suffice exactly when
the positive eigenangles sum to $4\pi$.
The example $\ell^{SU}(e^{2\pi i/N}I_N)=N-1$ shows that the determinant constraint can require
more than one factor beyond the rank bound.

The spectral characterizations (Theorems~\ref{thm:spectral} and~\ref{thm:spectral_U2}) place these factorizations
within the dictionaries defined by faithful $SU(2)$ and $U(2)$ embeddings.
For $n\ge2$, the phase-free dictionary is a finite union of compact sectors, whereas the phase-inclusive dictionary
is a dense, non-closed union of compact sectors with unbounded determinant twists.
Both are Haar-null. Distances to the phase-free dictionary and to fixed phase-inclusive sectors are
eigenvalue-matching problems (Theorem~\ref{thm:matching}, Corollary~\ref{cor:closed_forms});
the distance to the full phase-inclusive dictionary vanishes for every target.
Its spectral characterization does not by itself establish a terminating membership algorithm.

The geometric results identify the extremal embedding strata, the Hilbert--Schmidt scaling of embedded subgroups
and the threshold $\pi/d_{\max}(\typeSU)$ for ambient minimality of their geodesics
(Propositions~\ref{prop:extremal_strata}, \ref{prop:dynkin} and~\ref{prop:ambient_minimizing}).
Spectral spread also gives a first-order sharp bound on coherence production (Proposition~\ref{prop:coherence}).
These intrinsic statements do not determine tensor locality or entangling power, and the factor count does not include
the cost of realizing the chosen supports in a circuit.
Finite universal alphabets require additional choices; Proposition~\ref{prop:covariance} supplies monomial covariance
rather than full conjugation invariance.

Further questions concern the realizations and implementation of these intrinsic results.
\begin{enumerate}
\item[(1)] \emph{Geometry and realization of optimal factorizations.}
Describe the families of supporting planes and frames that realize the minimal factor count,
and determine how selecting among these factorizations affects implementation costs under tensor, locality or control constraints.
The exact intrinsic count does not include the cost of realizing the selected supports.
\item[(2)] \emph{Covariance beyond monomial groups.}
Determine which finite subgroups $\Gamma\subset U(N)$ admit universal alphabets
$\mathcal A(\mathcal W,\mathbf f;G_2):=\bigcup_{W\in\mathcal W}\phi_{W,f_W}(G_2)$
constructed from a $\Gamma$-orbit $\mathcal W$ of two-planes with compatible frames
$\mathbf f=(f_W)_{W\in\mathcal W}$, and obtain the corresponding compilation bounds.
Clifford-group orbits provide a natural case to investigate.
The monomial construction is one example, not a general basis-independent solution.
\item[(3)] \emph{Penalized geometries and implementation costs.}
Formulate and solve analogous distance and factorization-cost problems for right-invariant metrics that penalize selected directions.
Such metrics incorporate tensor and control structure and generally remove the conjugation invariance used by the spectral methods.
\end{enumerate}

\paragraph{Acknowledgments}
We are grateful to the anonymous referee, who pointed out the spectral description of the phase-free dictionary underlying
Theorem~\ref{thm:spectral} and asked how far a gate is from the dictionaries and which strata are extremal; these questions are
addressed in Sections~\ref{sec:strata} and~\ref{sec:distance}.

\section*{Statements and Declarations}

\paragraph{Funding}
This work was supported by the Generalitat Valenciana under grant COMCUANTICA/007 (QUANTWin), by the Agreement between the
Directorate-General for Innovation of the Ministry of Innovation, Industry, Trade and Tourism of the Generalitat Valenciana and the
Universidad CEU Cardenal Herrera, and by Universidad CEU Cardenal Herrera under grants INDI25/17 and GIR25/14.

\paragraph{Competing interests}
None

\paragraph{Data availability}
All results of this article are analytical. No datasets were generated or analysed during the current study; data sharing is not
applicable.

\paragraph{Code availability}
Short scripts used only as numerical sanity checks of some formulas are available from the corresponding author on reasonable request;
no result of the article depends on them.

\paragraph{Ethical standards}
The research meets all ethical guidelines, including adherence to the legal requirements of the study country.

\paragraph{Author contributions}
All authors contributed equally to this work. All authors jointly designed the study, developed the methodology, and drafted the
manuscript. All authors revised the manuscript critically and approved the final version.

\appendix
\section{Classical compilation, with explicit constants}
\label{app:compilation}

This appendix recalls classical results in the form used in the paper; see \cite[\S\S4.5.1--4.5.4]{NielsenChuang2010} for two-level
universality, error accumulation and Solovay--Kitaev compilation, \cite{ReckZeilingerBernsteinBertani1994} for the optical version and
\cite{Vartiainen2004,Mottonen2004} for gate-count optimized two-level decompositions. The numerical-linear-algebra references
\cite{GolubVanLoan2013,HornJohnson2013} are cited for the complex Givens rotation and the unitary QR algorithm. As explained in
Remark~\ref{rem:finite_extrinsic}, the intrinsic layer plays no role here beyond fixing the vocabulary: the alphabets are built on
coordinate planes and the compilation is the standard one.

\begin{lemma}[Diagonal tori]\label{lem:tori}
For $2\le j\le N$ and $t\in\R$ let $\gamma_{1j}(t):=\phi_{1,j}\big(\diag(e^{it/2},e^{-it/2})\big)$. Every $D_0\in T_0$ is a product
$D_0=\prod_{j=2}^N\gamma_{1j}(t_j)$; every $D\in T$ is a product of $N$ elements $\Phi_{1,j}(\diag(1,e^{is}))$, $\Phi_{1,2}(\diag(e^{is},1))$.
\end{lemma}

\begin{proof}
Write $D_0=\diag(e^{i\beta_1},\dots,e^{i\beta_N})$ with $\sum_j\beta_j=0$ and put $t_j:=-2\beta_j$ for $j\ge2$; the first diagonal
entry of the product is $e^{i\sum_jt_j/2}=e^{i\beta_1}$. The second claim is obvious.
\end{proof}

\begin{theorem}[Two-level factorization {\cite[\S4.5.1]{NielsenChuang2010}}]\label{thm:givens}
Every $U\in U(N)$ can be written as $U=\big(\prod_{k=1}^K\phi_{p_k,q_k}(S_k)\big)D$ with $S_k\in SU(2)$, $D\in T$, $\det D=\det U$ and
$K\le N(N-1)/2$.
\end{theorem}

\begin{proof}
To annihilate the entry $b$ below the entry $a$ in a column, left-multiply by $\phi_{p,q}(G)$, where $G=I_2$ if
$r:=(|a|^2+|b|^2)^{1/2}=0$ and $G=r^{-1}\begin{pmatrix}\bar a&\bar b\\-b&a\end{pmatrix}$ otherwise; then $G\in SU(2)$ and
$G(a,b)^{\top}=(r,0)^\top$. After at most $N(N-1)/2$ such steps one obtains an upper-triangular unitary, i.e.\ a diagonal $D$. Inverting
gives the claim, and $\det D=\det U$ because the $\phi_{p,q}(S_k)$ have determinant $1$.
\end{proof}

\begin{corollary}[Universality]\label{cor:universality_elem}
$\langle\Gtl(n)\rangle=\langle\Gel(n)\rangle=SU(N)$ and $\langle\GtlU(n)\rangle=\langle\GelU(n)\rangle=U(N)$.
\end{corollary}

\begin{proof}
For $U\in SU(N)$, Theorem~\ref{thm:givens} gives $D\in T_0$, which is a product of elements of $\Gtl$ by Lemma~\ref{lem:tori}. For
$U\in U(N)$, $D$ is a product of elements of $\GtlU$. The statements for $\Gel$, $\GelU$ follow from
$\Gtl\subset\Gel\subset SU(N)$ (by \eqref{eq:Gel_in_SU}) and $\GtlU\subset\GelU\subset U(N)$.
\end{proof}

\paragraph{Solovay--Kitaev.}
Let $G_2\subset SU(2)$ be finite with $\langle G_2\rangle$ dense in $SU(2)$, and let $\mathcal A(n;G_2)$ be the coordinate alphabet of
Section~\ref{sec:distance}, with words as defined there. By the Solovay--Kitaev theorem
\cite{Kitaev2002,DawsonNielsen2006}, \cite[App.~3]{NielsenChuang2010} there are constants $c,C>0$, depending only on $G_2$, and an
algorithm that returns, for every $V\in SU(2)$ and $\delta\in(0,1)$, a word $w$ over $G_2\cup G_2^{-1}$ with $\|V-w\|\le\delta$ and
$|w|\le C\big(1+\log(1/\delta)\big)^c$. (The usual form $C\log^c(1/\delta)$ holds for $\delta\le\delta_0<1$; for $\delta\in(\delta_0,1)$
one uses the word for $\delta_0$ and enlarges $C$.) Better exponents are available \cite{Kuperberg2025BreakingCubicSK}. Two elementary facts are used:
$\|\phi_{p,q}(V)-\phi_{p,q}(V')\|=\|V-V'\|$ (block-diagonal structure), and, for unitaries,
\begin{equation}\label{eq:telescoping}
\Big\|\prod_{k=1}^KA_k-\prod_{k=1}^KB_k\Big\|\le\sum_{k=1}^K\|A_k-B_k\|
\end{equation}
\cite[\S4.5.3]{NielsenChuang2010}, by telescoping.

\begin{theorem}[Compilation with explicit constants]\label{thm:compilation}
Let $U\in U(N)$, $\varepsilon\in(0,1)$, and let $U=\big(\prod_{k=1}^K\phi_{p_k,q_k}(S_k)\big)D$ be as in Theorem~\ref{thm:givens}. Then
there is a word $W$ over $\mathcal A(n;G_2)\cup\mathcal A(n;G_2)^{-1}$ with
\[
\|U-WD\|\le\varepsilon,\qquad |W|\le K\,C\Big(1+\log\frac{\max\{1,K\}}{\varepsilon}\Big)^c,\qquad K\le N(N-1)/2.
\]
\end{theorem}

\begin{proof}
If $K=0$, then $U=D$ and the empty word gives error $0$. If $K\ge1$, approximate each $S_k$ by a Solovay--Kitaev word $w_k$ with
$\|S_k-w_k\|\le\varepsilon/K$, lift letterwise through $\phi_{p_k,q_k}$, set $W:=\prod_k\phi_{p_k,q_k}(w_k)$, and apply
\eqref{eq:telescoping} and the unitary invariance of $\|\cdot\|$.
\end{proof}

\begin{corollary}[Pure-word compilation up to global phase]\label{cor:pure_word}
Let $U\in U(N)$, $\varepsilon\in(0,1)$, $e^{i\theta}:=(\det U)^{1/N}$ (any $N$-th root) and $K$ as in Theorem~\ref{thm:givens}. Then there
is a word $\widetilde W$ over $\mathcal A(n;G_2)\cup\mathcal A(n;G_2)^{-1}$ with
\[
\|U-e^{i\theta}\widetilde W\|\le\varepsilon,
\qquad
|\widetilde W|\le (K+N-1)\,C\Big(1+\log\frac{K+N-1}{\varepsilon}\Big)^c.
\]
\end{corollary}

\begin{proof}
In Theorem~\ref{thm:givens}, $\det(e^{-i\theta}D)=e^{-iN\theta}\det U=1$, so $e^{-i\theta}D\in T_0$ and, by Lemma~\ref{lem:tori},
$e^{-i\theta}D=\prod_{j=2}^N\phi_{1,j}(g_j)$ with $g_j\in SU(2)$. Hence $U=e^{i\theta}\prod_{k=1}^{K+N-1}\phi_{p_k,q_k}(S_k)$ with
$S_k\in SU(2)$ for all $K+N-1$ factors. Approximate each $S_k$ with accuracy $\varepsilon/(K+N-1)$ and conclude as in
Theorem~\ref{thm:compilation}.
\end{proof}

\begin{remark}
The diagonal factors $\phi_{1,j}(g_j)$ are elements of continuous one-parameter subgroups, not letters of the finite alphabet; they
must be approximated as well, which is the origin of the term $N-1$ in Corollary~\ref{cor:pure_word}. In the worst case the number of
continuous two-level factors is $O(N^2)=O(4^n)$, reflecting $\dim U(2^n)=4^n$, and the compiled word length is
$O\big(N^2(1+\log(N^2/\varepsilon))^c\big)$. Any approximation routine on $SU(2)$ over a prescribed alphabet (for
instance Clifford$+T$ synthesis) can replace Solovay--Kitaev, with the same lifting and error analysis.
\end{remark}

\end{document}